\documentclass[a4paper,babel]{article}
\usepackage{cite}
\usepackage{amsmath,amssymb,amsfonts,amsthm}
\usepackage{graphicx}

\usepackage{enumerate}
\usepackage{mathtools}
\usepackage[dvipsnames]{xcolor}
\mathtoolsset{showonlyrefs,showmanualtags}

\usepackage{a4wide}

\usepackage{enumerate}

\newcommand{\id}{\boldsymbol{I}}

\newcommand{\bu}{\boldsymbol{u}}

\newcommand{\bx}{\boldsymbol{x}}
\newcommand{\bb}{\boldsymbol{b}}

\newcommand{\llangle}{\langle\!\langle}
\newcommand{\rrangle}{\rangle\!\rangle}
\newcommand{\brho}{\boldsymbol{\rho}}

\newcommand{\Hdue}[1]{h_{#1}}

\newcommand{\Hse}[1]{M_{#1}}

\newcommand{\proj}{\mathbb{P}}

\newcommand{\pur}[1]{ {#1} {#1}^{\dagger}}
\newcommand{\cappa}{\kappa}
\newcommand{\bV}{\mathcal{W}}

\newcommand{\probbo}{\textbf{(P)}}
\newcommand{\frulli}{\omega}

\newtheorem{remark}{Remark}
\newtheorem{proposition}{Proposition}
\newtheorem{definition}{Definition}

\newtheorem{lem}{Lemma}
\newtheorem{theorem}{Theorem}

\begin{document}

\title{Indirect stabilization of $N$-level quantum systems}
\author{Francesca Carlotta Chittaro
\thanks{Francesca C. Chittaro is with  Department of
Information Engineering and Computer Science, Università di Trento, 38123 Trento, Italy
  (e-mail: francesca.chittaro@unitn.it).}}
\date{}

\maketitle

\begin{abstract}
We study the stabilizability of a $N$-level quantum systems using an indirect control architecture mediated by a noisy ancilla. After outlining the general theoretical framework for arbitrary $N$, we focus on the two-level case ($N=2$), for which we derive a complete set of necessary and sufficient conditions for asymptotic stability. Finally, we illustrate the analytical results with a simple numerical example on a two-spin system, showing convergence to a Bell state for a set of random initial conditions.
\end{abstract}

\section{Introduction}

Dissipation engineering \cite{reservoir} has fundamentally shifted the paradigm of open quantum systems control by turning environmental noise into a resource. By properly designing a system's interaction with its surroundings, environmental dissipation can be engineered to perform tasks such as autonomous state stabilization, rapid state initialization, and entanglement generation (see, for instance, \cite{MiLeVo13,PhysRevLett.119.150502,ShaHaLeal13,HaMuMu22,SchirmerWang10,TiSchiWa-TAC}). 

Within this framework, this paper provides a theoretical characterization of static state stabilization via a noisy ancilla. Specifically, we consider an unmonitored bipartite configuration where the system of interest is coupled to a (possibly controlled) dissipative ancillary qubit. In this setting, both the coherent control actions and the Markovian environmental dissipation are directly coupled to the ancilla only. Our aim consists in determining the structural conditions under which the localized ancilla-system interaction can successfully drive the system of interest toward a desired pure state, without performing active measurements or feedback loops.

Our approach relies on a geometric analysis of the equilibrium set of the open quantum dynamics, by studying how the ancilla-system coupling shapes the fixed-point set and deriving explicit structural constraints that guarantee the existence of an attractive target equilibrium.

Building upon the preliminary results in \cite{C26}, which were restricted to qubit systems and target states aligned with the internal system Hamiltonian, this paper addresses the problem in a more general setting. We consider target systems of arbitrary finite dimension and arbitrary pure target states, deriving general necessary conditions for indirect stabilization.

Stabilizing a system toward an equilibrium is a central problem in both classical and quantum systems theory. Specifically in the quantum realm, 
the problem of stabilizing the state of subsystems is crucial. From a different viewpoint, this problem has been addressed in a series of paper by Ticozzi and collaborators \cite{TiVi08,TiVi09,TiSchiWa-TAC,TiLuCapVi12}. In these papers, the focus was on the stabilization of states supported by {\it subspaces} of the Hilbert space, instead of the stabilization of reduced states. Another key difference resides in the fact that the authors allow the use of feedback control, while here we consider static control only and we do not measure the system.

For results in infinite dimensions, we cite~\cite{RoRoSeIHP}, which investigates state stabilization (via reservoir engineering) for a system of two bosonic modes.

%
%
%
The structure of the paper is the following: in Section~\ref{sec: system} we define the mathematical framework and state the main hypotheses; in Section~\ref{sec: coherence rep}, we introduce the coherence vector representation, a standard algebraic tool that permits to write the evolution of the quantum system as a bilinear system on the real Euclidean space of appropriate dimension. 

Section~\ref{sec: any dim} contains some of the main results of the paper: in particular, we prove that, regardless the dimension of the system, asymptotic stabilization towards a pure state is possible if and only if the dynamics admits a unique equilibrium, which is shown to be a pure factorized state of the composite system. Further,
we establish explicit necessary conditions on the ancilla-system coupling for the existence of such an equilibrium. 

Because the classification of equilibria for generic $N$ becomes analytically intractable, due to the rapid growth of the dimension, Section~\ref{sec: qubit} restricts the analysis to two-level systems ($N=2$). 
In this case,  we are able to 
provide a set of necessary and sufficient conditions for the asymptotic stabilization of the system towards the target state.

In Section~\ref{sec: example}, we illustrate our result by a numerical example, in which we aim at stabilizing the system of interest, composed of two qubits, towards a Bell state.

\section{Statement of the problem} \label{sec: preliminaries}
\label{sec: system}

We consider the following composite system: a finite dimensional quantum system $\bf S$ interacts with a qubit $\bf A$, called {\it ancilla}, which is subject to dissipation (interaction with the environment) and on which acts a coherent control \cite{altaf-tic-12}.
The state of such a system is described by a density matrix
$\brho$ (i.e., a Hermitian positive semi-definite matrix of trace 1) on the tensor product 
$\mathcal{H} = \mathbb{C}^2 \otimes \mathbb{C}^N$ ($N$ denotes the dimension of \textbf{S}). 
We denote with $\mathcal{P}$ the space of density matrices on $\mathcal{H}$. It is well known \cite{nielsen-chuang} that $\mathcal{P}$ is a compact convex subset of the space of Hermitian matrices on $\mathcal{H}$. The constraints on $\brho$ impose that the trace of $\brho^2$ is always less than or equal to 1. States for which $\brho^2$ has unit trace are of utmost importance \cite{nielsen-chuang} and are called {\it pure states}.

The state $\rho_S$ of the system of interest \textbf{S} (respectively, the state $\rho_A$ of the ancilla) can be recovered from $\brho$ by means of the {\it  partial trace  over $\bf A$} (respectively, the partial trace over $\bf S$) 
defined as follows: for every $m\geq 2$, let $her(m)$ denote the space of Hermitian $m$-dimensional matrices; the partial trace of {\bf A} is 
the  unique linear operator  
tr$_A: her(2N) \to her(N)$  such that every $M_A\in her(2)$ and every $M_S\in her(N)$ it holds tr$_A (M_A\otimes M_S)=M_S({\rm tr} M_A)$, where tr denotes the usual trace operation on $her(2)$. The partial trace over {\bf S} is defined analogously (see \cite{nielsen-chuang}). 

A rigorous description of the dynamics of the composite system should also take into account its interaction with the surrounding environment, usually modeled as an infinite collection of bosons \cite{breuer-petr} (but other models are available). 
Such problem belongs to the realm of Mathematical Physics and goes beyond the scope of this work\footnote{For a similar stabilization problem within this fundamental approach see~\cite{FroSchu16}.}. 
Here, we adopt the standard effective framework provided by the Markovian Gorini-Kossakowski-Sudarshan-Lindblad
 master equation
(or GKSL, see~\cite{gorini-kos-sud,lindblad}): the dynamics of the systems is governed by the equation  
\begin{equation} \label{eq: GKSL}
\dot{\brho}(t) = -i [H, \brho(t)] + \mathcal{D}(\brho(t)),
\end{equation}
where $H$ is a 
Hermitian operator on $\mathcal{H}$ and the linear operator $\mathcal{D}$ can be written as 
\begin{equation}\label{eq: lind}
\mathcal{D}(\brho) = \sum_{k=1}^{m}\mathcal{D}_{L_k}(\brho)= \sum_{k=1}^{m} L_k \brho L^{\dagger}_k - \frac{1}{2} L^{\dagger}_k  L_k  \brho- \frac{1}{2} \brho L^{\dagger}_k L_k,
\end{equation}
the $L_k$ being square matrices called {\it noise (or jump) operators} and $m\leq (N+2)^2-1$.

It is well known that different choices of the operators $(H,L_1,\ldots,L_{m})$ can generate the same evolution semigroup. First, adding any multiple of the identity to $H$  leaves the right-hand side of \eqref{eq: GKSL} invariant. Analogously, as $\mathcal{D}_{\alpha F}=|\alpha|^2\mathcal{D}_{F}$, multiplying any noise operator by a unit-norm complex number does not change the right-hand side of \eqref{eq: GKSL}. Then, notice that, for every operator $L$ and every $\alpha\in \mathbb{C}$, we have
\[\mathcal{D}_{\alpha \id + L}=\mathcal{D}_L-i\Big[\frac{i}{2}(\alpha L^{\dagger}-\alpha^*L),\cdot\Big],\] 
so that the (Hermitian) term $\frac{i}{2}(\alpha L^{\dagger}-\alpha^*L)$ can be reabsorbed into $H$. 
Finally, even fixing the Hamiltonian, different sets of noise operators can generate the same $\mathcal{D}$ 
(for more details, see~\cite{alicki-fannes,breuer-petr}).

\smallskip
In the case under study, the assumptions that the subsystem $\bf S$ is not  directly coupled neither with the environment nor with the control translate in the following mathematical properties:
\begin{itemize}
 \item[(D)]  
 the noise operators can be taken of the form $L_k=\ell_k\otimes\id_N$, where the $\ell_k$'s are 2-dimensional complex matrices, $k\leq 3$, and $\id_N$ denotes the $N$-dimensional unity matrix.
 \item[(H1)] 
 the Hamiltonian $H$ is an affine function of the controls and can be written as 
\[
H(\bu)=H_0+\sum_{j=1}^{m_c} u_j \Hdue{A_j}\otimes \id_N,
\]
where $H_0$  is a constant $2N$-dimensional Hermitian matrix, $ \Hdue{A_j}$ are constant 2-dimensional Hermitian matrices and $u_j$ are real (possibly time-dependent) parameters.
\end{itemize}

\begin{remark}
{\it Coherent control} refers to the situation in which one can manipulate the state of the system by applying semiclassical potentials \cite{altaf-tic-12}. Such potentials are described by the Hamiltonians $H_j$, $j\geq 1$.
\end{remark}
  
Exploiting the properties of GKSL equation, we  can assume that $H_0$ and $h_{A_j}$, $j\geq 1$, are traceless. Also, as the space of traceless Hermitian 2-dimensional matrices has dimension 3, we can assume that 
we have only three controls.  
To fix notations, we write the Hamiltonian $H_0$ as 
\[
H_0=h_A\otimes \id_N + H_I + \id_2\otimes h_S,
\]
where ${\rm tr}_A H_I={\rm tr}_SH_I=0$.
$h_A$ denotes the internal evolution of the ancilla, $h_S$ the evolution of the $N$-level system {\bf S}, and $H_I$ represents the interaction between the ancilla and the system of interest (sometimes we will refer to $H_I$ as to the {\it coupling}). Notice that, for every fixed $H_0$, this decomposition is unique.

We can also simplify the expression of the noise operators: indeed, in \cite[Proposition~9]{C26} it has been shown that the system {\bf S} cannot be {\it globally} asymptotically stabilized to the target set if the matrices $\ell_j$, $j=1,2,3$ are not simultaneously triangularizable. Then, we assume that there exists a change of basis in $\mathbb{C}^2$ such that, for every $j=1,2,3$, $\ell_j$ has the form 
\begin{equation} \label{eq: jump triang}
\ell_j=\begin{pmatrix} a_j & b_j e^{-i\chi_j}\\ 0 & -a_j\end{pmatrix},\end{equation}
where $a_j,b_j$ and $\chi_j$ are real parameters.  

We are now ready to formally state the stabilization problem under study. 

\smallskip
\noindent
\probbo\ {\it 
Let $h_A,h_S,h_{A_j}$ and $\ell_j$ be fixed. 
Given a unit-norm $\psi_*\in \mathbb{C}^N$, is it possible to design a coupling $H_I$ and to choose fixed values for the controls $u_j$ such that all solutions of \eqref{eq: GKSL} are asymptotically driven towards the set
$\mathcal{W}_{\psi_*}=\{\brho: {\rm tr}_{A}\rho=\psi^{\dagger}_*\psi_*\}$?}

\subsection{Notations}

When dealing with quantum mechanics, it is mostly convenient to adopt the {\it bra-ket} notation. In this framework, {\it kets} are denoted with $|\psi \rangle$ and correspond to vectors ($\psi$) of the Hilbert space; the corresponding {\it bra} $\langle \psi|$ just denotes the conjugate transpose of $\psi$, i.e. $\psi^{\dagger}$. 
When $\psi$ is a unit vector, $|\psi\rangle \langle \psi|$ is simply the orthogonal projection on the span of $\psi$. If $e_i,e_j$, for $1\leq i\leq j\leq n$, are  elements of the canonical basis of $\mathbb{C}^n$, then $|e_i\rangle \langle e_j|$ denotes the matrix $M$ whose  only non-zero element is $M_{ij}=1$.

\section{Coherence vector representation} \label{sec: coherence rep}

Originally developed for qubit systems, and known as {\it Bloch representation},
coherence vector representation  
is an isomorphism from the set of density matrices (on $\mathbb{C}^n$, for some $n\geq 2$) to a suitable subset of $\mathbb{R}^{n^2}$. 

To define such isomorphism, we start by endowing the space of $n$ dimensional Hermitian matrices $her(n)$ with the Frobenius scalar product $\llangle A,B\rrangle={\rm tr}\big(A^{\dagger}B\big)$, then we choose a orthonormal basis $\{\Lambda_0,\ldots,\Lambda_{n^2-1}\}$ of $her(n)$.  
As we focus on density matrices (i.e., with unit trace), a common choice is setting $\Lambda_0=\id_n/\sqrt{n}$, and choosing the other $n^2-1$ elements as traceless Hermitian matrices. 

Once the basis fixed,
the coherence representation is given by the map $\Phi$ defined as  
\begin{align}
\Phi: \mathcal{P}&\to \mathbb{R}^{n^2}\\
\brho &\mapsto\overline{\bx^{\rho}}=\big({\rm tr}(\brho \Lambda_0),\ldots,{\rm tr}(\brho \Lambda_{n^2-1})\big).
\end{align}
As ${\rm tr}(\brho \Lambda_0)=1/\sqrt{n}$ for every $\brho$, the image of $\mathcal{P}$ under $\Phi$ lies in the affine subspace of $\mathbb{R}^{n^2}$ of vectors whose first component is $1/\sqrt{n}$. More precisely, $\Phi(\mathcal{P})$ is a compact convex subset of such affine space. We recall that the boundary of $\Phi(\mathcal{P})$ contains the vectors associated with singular density matrices, and that the extreme points of the convex set $\Phi(\mathcal{P})$ are precisely the vectors corresponding to pure states \cite{avron}.

For 1 qubit ($n=2$), the natural choice 
is $\Lambda_j=\sigma_j/\sqrt{2}$, $j=1,2,3$, where $\sigma_j$ are
the Pauli matrices
\[
\sigma_1=\left(\begin{smallmatrix}
               0 & 1 \\ 1 & 0
              \end{smallmatrix}\right)\quad 
              \sigma_2=\left(\begin{smallmatrix}
               0 & -i \\ i & 0
              \end{smallmatrix}\right)\quad
\sigma_3=\left(\begin{smallmatrix}
               1 & 0 \\ 0 & -1
              \end{smallmatrix}\right).              
\]
As the first component of $\overline{\bx^{\rho}}$ is always $1/\sqrt{2}$, one usually writes $\overline{\bx^{\rho}}=\big(1/\sqrt{2},\bx^{\rho}\big)$, where $\bx^{\rho}$ is called the
{\it Bloch vector}.
 It is easy to see that $\|\overline{\bx^{\rho}}\|^2={\rm tr}\brho^2$, so that $\bx^{\rho}$ belongs to the ball of radius $1/\sqrt{2}$ (the {\it Bloch ball}). In particular, a Bloch vector represents a pure state if and only if $\|\bx^{\rho}\|^2=1/2$~\footnote{Other normalization choice are also adopted, in which the Bloch ball coincides with the unit ball in $\mathbb{R}^3$.}. 

Actually, the equality $\|\overline{\bx^{\rho}}\|^2={\rm tr}\brho^2$ holds true for every $n$. 
Choosing $\Lambda_0=\id_n/\sqrt{n}$ and $\Lambda_j$ traceless for $j\geq 1$ allows us to write the $n^2$-dimensional vector $\overline{\bx^{\rho}}$ can as $\overline{\bx^{\rho}}=\big(1/\sqrt{n},\bx^{\rho}\big)$, where $\bx^{\rho}$
is called the {\it vector of coherences} of $\brho$. 
Thus, one always has $\|\bx^{\rho}\|^2\leq 1-1/n$, where the inequality is saturated for pure states only. It is worth noticing that, for $n>2$, the set of $\bx^{\rho}\in \mathbb{R}^{n^2-1}$ such that $\overline{\bx^{\rho}}\in \Phi(\mathcal{P})$ is strictly contained in the ball of radius $1-1/n$. 
%

For $n>2$, 
one can choose the matrices $\Lambda_j$, $j\geq 1$, as the standard generators of the Lie algebra $\mathfrak{su}(n)$, suitably normalized (see Appendix~\ref{sec: su(N)}). However,
in the bipartite system under study, a convenient choice of the orthonormal basis would respect the tensor product structure of the Hilbert space. For this reason,  we construct the basis $\{\Lambda_0,\ldots,\Lambda_{(2N)^2-1}\}$  as follows. We set:
\begin{gather}
\Lambda_0=\frac{1}{\sqrt{2N}}\id_{2N},\\
\Lambda_i=\frac{1}{\sqrt{2N}}\sigma_i\otimes \id_{N},\quad i=1,2,3, 
\\ 
\Lambda_{3+(N^2-1)(i-1)+j}=\frac{\sigma_i}{\sqrt{2}}\otimes S_j,\qquad  \begin{array}{c} i=1,2,3,\\ j=1,\ldots,N^2-1,\end{array}\\
\Lambda_{3N^2+j}=\frac{\id_{2}}{\sqrt{2}}\otimes S_j,\qquad j=1,\ldots,N^2-1.
\end{gather}
where the matrices $S_j$ (multiples of the standard generators of $\mathfrak{su}(N)$) are available in Appendix~\ref{sec: su(N)}. 

Thanks to this choice of the basis, the vector of coherences acquires a form reflecting the bipartite structure of the system: indeed, 
setting
\[
\boldsymbol{x^A}=\left(\begin{smallmatrix}
        \bx^{\rho}_1\\\bx^{\rho}_2\\\bx^{\rho}_3               
                      \end{smallmatrix}
\right)\quad 
\boldsymbol{x^{AS}}=\left(\begin{smallmatrix}
        \bx^{\rho}_4\\\vdots\\\bx^{\rho}_{3N^2}               
                      \end{smallmatrix}
\right)\quad 
\boldsymbol{x^S}=\left(\begin{smallmatrix}
        \bx^{\rho}_{3N^2+1}\\\vdots\\\bx^{\rho}_{4N^2-1}               
                      \end{smallmatrix}
\right),
\]
we have that  $\big(\frac{1}{\sqrt{2}},\sqrt{N}\boldsymbol{x^A}\big)$ is the Bloch vector associated with ${\rm tr}_S\brho$ and $\big(\frac{1}{\sqrt{N}},\sqrt{2}\boldsymbol{x^S}\big)$ is the coherence vector representation
of ${\rm tr}_A\brho$. In view of this, the partial trace operations act on $\Phi(\mathcal{P})$ as the composition of a product by a normalizing factor and a projection.

We end this section by showing the particular structure of the coherence vectors associated with  factorized states. Assume that  there exist a 2-dimensional density matrix $\rho_A$ and a $N$-dimensional density matrix $\rho_S$ such that $\brho=\rho_A\otimes\rho_S$.  Easy computations show that the coherence vector $\boldsymbol{x}^{\rho}$ associated with $\brho$ has the form
\begin{equation} \label{eq: X factor}
 \boldsymbol{x}^{\rho}=\big(\boldsymbol{x^A},\sqrt{2N}\boldsymbol{x^A}\otimes \boldsymbol{x^S},\boldsymbol{x^S}\big).
\end{equation}

\subsection{Coherence representation of the GKSL equation}

 Taking advantage of the basis $\{\Lambda_1,\ldots,\Lambda_{4N^2-1}\}$ above defined, we can write
\begin{gather}
h_A=\sum_{i=1}^3 \frac{\alpha_i}{2}\sigma_i,\qquad h_S=\sum_{j=1}^{N^2-1}\frulli_j S_j\\ 
H_I=\sum_{i=1}^3\sum_{j=1}^{N^2-1}\sqrt{2}h_{ij}\sigma_i\otimes S_j, \label{eq: HI}
\end{gather} 
for some real values $\alpha_i,\frulli_j,h_{ij}$.
Possibly changing the coordinates in the space of controls, we can assume that 
\[
H(\bu)-H_0=\sum_{i=1}^3 \frac{u_i}{2} \sigma_i\otimes 
\id_N.
\]
As, in this paper, the control are kept constant, we can rename the sums $\alpha_i+u_i$ as $\alpha_i$, $i=1,2,3$ and assume, in the following, that we can possibly choose the value of the frequencies $\alpha_i$ (actually, what we really will need is to possibly slightly change the value of $\alpha_i$, if it is zero or if it coincides with come particular function of other parameters).

In coherence  vector representation,  equation~\eqref{eq: GKSL} becomes 
\begin{equation} \label{eq: coherent controlled GKSL}
\dot{\overline{\boldsymbol{x^{\rho}}}}(t)=\big(\Hse{H} 
+\Hse{\mathcal{D}}\big)\overline{\boldsymbol{x^{\rho}}}(t),
\end{equation}
$\Hse{H}$  and $\Hse{\mathcal{D}}$ being respectively the representations of the operators $-i[H(\bu),\cdot]$
and $\mathcal{D}(\cdot)$, and they are real matrices.

 In particular, $M_H$ is an antisymmetric matrix having the block form
\begin{align}
 \Hse{H}&=\left(
\begin{smallmatrix}
 0 &  {\bf 0} &  {\bf 0} & {\bf 0}\\
 & & & \\
 {\bf 0} & \widehat{ {H}}_{A} & \widehat{ {H}}_{It} &  {\bf 0}\\
  & & & \\
  & & & \\
  {\bf 0} & \widehat{ {H}}_{Il}& \big(\widehat{ {H}}_{A}\otimes  {\id_{N^2-1}} +  {\id_{3}}\otimes \widehat{ {H}}_{S}+H_{Ic}\big) & \widehat{ {H}}_{Ir}\\
  & & & \\
  & & & \\
  {\bf 0} &  {\bf 0}& \widehat{ {H}}_{Ib} &  \widehat{ {H}}_{S} \\
  & & & 
\end{smallmatrix}
\right),
\end{align}
where $\widehat{ {H}}_{A}$ is square of dimension 3, $\widehat{ {H}}_{S}$ is square of dimension $N^2-1$, and $\widehat{ {H}}_{1c}$ is square of dimension $3(N^2-1)$. More precisely, $\widehat{ {H}}_{A}=\sum_{i=i}^3 \alpha_i T_i$,  
where the matrices $T_i$ are the adjoint representation of the Pauli matrices (up to a normalization) and are given by
\[
T_1=\left(\begin{smallmatrix}
      0 & 0 & 0\\
            0 & 0& -1\\
                  0 & 1 & 0
     \end{smallmatrix}\right)\quad
T_2=\left(\begin{smallmatrix}
      0 & 0 & 1\\
            0 & 0& 0\\
                  -1 & 0 & 0
     \end{smallmatrix}\right)\quad 
     T_3=\left(\begin{smallmatrix}
      0 & -1 & 0\\
            1 & 0& 0\\
                  0 & 0 & 0
     \end{smallmatrix}\right).
\]
The expressions of the other blocks can be written in terms of the\ families of matrices $C_j$ and $D_j$ defined in Appendix~\ref{sec: su(N)}: 
\begin{align}
\widehat{H}_S&=\sum_{j=1}^{N^2-1}\frulli_jC_j^T\\
H_{It} &= 2\sqrt{\frac{2}{N}}\sum_{l=1}^3 T_l\otimes \big(h_{l1},\ldots,h_{l N^2-1}\big)\\
H_{Ib} &= \sqrt{2}\sum_{j=1}^{N^2-1} \left(h_{1j}, h_{2j},h_{3j} \right) \otimes C^T_j\\
H_{Ic} &= \sqrt{2}\sum_{l=1}^3\sum_{j=1}^{N^2-1} h_{lj} T_l \otimes D^T_j.
\end{align}
By anti-symmetry, it holds $H_{Il}=-H_{It}^T$ and $H_{Ir}=-H_{Ib}^T$.

Concerning $\Hse{\mathcal{D}}$, by linearity we have that
$\Hse{\mathcal{D}}=\sum_{i=1}^3\Hse{\mathcal{D}_{\ell_i}}$, where each of the 
$\Hse{\mathcal{D}_{\ell_i}}$ 
has the form
\[
\Hse{\mathcal{D}_{\ell_i}}=\left(
\begin{smallmatrix}
 0 & {\bf 0} &   {\bf 0} &    {\bf 0} \\
 & & & \\
 \cappa_i & \Gamma & {\bf 0} &  {\bf 0}\\
  & & & \\
  & & & \\
  {\bf 0} & {\bf 0}& \Gamma_i\otimes  {\id_{N^2-1}} & \cappa_i\otimes \id_{N^2-1}\\
  & & & \\
  & & & \\
  {\bf 0} &  {\bf 0}& {\bf 0}&  {\bf 0} \\
  & & & 
\end{smallmatrix}
\right),
\]
with 
\begin{gather} \label{eq: gammai}
\Gamma_i=\left(
\begin{smallmatrix}
-2a_i^2-\frac{1}{2}b_i^2 & 0 & a_ib_i\cos(\chi_i)\\
0 & -2a_i^2-\frac{1}{2}b_i^2 & a_ib_i \sin(\chi_i)\\
a_ib_i \cos(\chi_i) & a_ib_i \sin(\chi_i) & -b_i^2
\end{smallmatrix}\right)\\
\kappa_i=\left(\begin{smallmatrix}
   -2a_ib_i \cos(\chi_i)\\
      -2a_ib_i \sin(\chi_i)\\
         b_i^2
 \end{smallmatrix}\right).\label{eq: kappai}
\end{gather}
In the following, we set $\Gamma_i=\sum_{i=1}^3 \Gamma_i$
and $\kappa_i=\sum_{i=1}^3 \kappa_i$.

Let us recall that, that, if all $\ell_i$'s are a normal matrices (which in the coordinates we chose corresponds to $b_i=0$ for every $i$), then the vector $\kappa$ is null  
(these are the so-called {\it unital} evolutions, see~\cite{altafiniJMP,dirr-helmke}). In systems governed by unital evolutions, the completely mixed state ($\brho=\id_{2N}/\sqrt{2N}$) is always an equilibrium of the system and the purity of the state is always non-increasing; remark moreover that such property depends on the dissipative part $\mathcal{D}$ only, and is completely independent on the Hamiltonian of the system.  Then, as it will be clear in the next Section, if the evolution is unital, the system cannot be globally stabilized towards the target set $\mathcal{W}_{\psi^*}$. The fact that at least one of the $\ell_i$ is not a normal operator is thus a necessary condition for stabilizability.


\section{Stabilizing any pure reduced state: necessary conditions} \label{sec: any dim}

\subsection{On the necessity of a factorized pure equilibrium}

Problem \probbo\ can be rigorously formulated as follows: given a unit-norm target state $|\psi_*\rangle \in \mathbb{C}^{N}$,
and setting $\mathcal{W}_{\psi_*}=\{\brho \in \mathcal{P}_{AS}: {\rm tr}_A\brho=|\psi_*\rangle\langle \psi_*|\}$
find a coupling such that, for every trajectory $\brho(t)$ solution of GKSL, we have 
\[
\lim_{t\to+\infty} {\rm dist}(\brho(t),\mathcal{W}_{\psi_*})=0,
\]
where dist denotes any distance on the set of density matrix.

A necessary condition for this to occur is that, for the chosen values of the couplings, there exists an invariant set of GKSL contained in $\mathcal{W}_{\psi_*}$.   In this section we will analyze the implication of such a requirement and come up with necessary conditions on the choice of the couplings.

To simplify the analysis, 
we apply on $\mathbb{C}^{2N}$ a coordinate change of the form $\id_2\otimes O$, with $O$ unitary of dimension $N$ such that  $O|\psi_*\rangle=(0,\ldots,0,1)^T$. In these coordinates, the target set becomes
\begin{equation} \label{eq: target set rho}
\mathcal{W}_{\psi_*}=\Bigg\{\brho \in \mathcal{P}: {\rm tr}_A \brho=\left(\begin{smallmatrix}
                            &  &  & 0\\
                            &   \text{\Large 0} && \vdots\\
                            & &   &0\\
                           0 &   \ldots   &0& 1
                          \end{smallmatrix}\right)
\Bigg\}.
\end{equation}
Let us remark that such coordinate change  does not affect the operators $h_A$ and $\ell$. 

It is a standard result (that can be verified by straight computation applying Schmidt decomposition~\cite{nielsen-chuang}) in quantum information theory that all elements in the target set \eqref{eq: target set rho} are {\it factorized}. More precisely, the following (more general) statement holds true.
\begin{lem} \label{lem: fact}
Let $\brho$ be a density matrix in the Hilbert space $\mathbb{C}^{n_A}\otimes \mathbb{C}^{n_B}$, with $n_A,n_B\geq 2$. Assume that tr$_A\brho$ is a pure state. Then there exists a density  matrix $\rho_A$ on $\mathbb{C}^{n_A}$ such that $\brho=\rho_A\otimes {\rm tr}_A\brho$.
\end{lem}
The Lemma holds true if we reverse the roles of the subsystems A and B.

As a consequence of Lemma~\ref{lem: fact}, 
the set $\Phi(\mathcal{W}_{\psi_*})$ reads 
\begin{equation} \label{eq: Vmeno}
V_{\psi_*}=\left\{\Big(\frac{1}{\sqrt{2N}},\bx^A,\sqrt{2N}\bx^A\otimes \bx^B_*,\bx^B_*\Big):
\begin{array}{c}
\bx^A\in \mathbb{R}^3,\\ \|\bx^A\|\leq 1/2\end{array}  \right\},
\end{equation} 
where 
\[
\bx^B_*=\Big(0,\ldots,0,-\sqrt{\frac{N-1}{2N}}\Big)\in  \mathbb{R}^{N^2-1}.
\]
For conciseness' sake, in the following we set $\bb_*=-\sqrt{\frac{N-1}{2N}}$.

Before carrying on the analysis, we make the following assumption: 
\begin{itemize}
 \item[(H2)] The set 
\[
\{\frulli_j,h_{1j},h_{2j},h_{3j}: 1\leq j < N^2-1, j\neq r^2-1, 1\leq r\leq N \}
\]
contains at least one non-zero element.

\end{itemize}

On the one hand, the assumption is legitimate, as, in our viewpoint, we can suitably choose the values of the couplings $h_{ij}$ (the $\omega_j$ are fixed, though). On the other hand, we remark that (H2) is  necessary  for stabilizability. 
Assume indeed, for instance, that $\frulli_j=0$ and $(h_{1j},h_{2j},h_{3j})=(0,0,0)$ for every $j\neq r^2-1$, for every $1\leq r\leq N$. If so, direct computations show that the last line of the matrix $M_H+M_D$ is zero, that is, if $\overline{\bx^{\rho}}(t)$ is any solution of \eqref{eq: coherent controlled GKSL}, its last component is constant in time. This prevents stabilization to the target set for states for which the last component of $\overline{\bx^{\rho}}(t)$ is not identically $\bb_*$. 

The main result of this section states that the only invariant set contained in $V_{\psi_*}$, if it exists, is an equilibrium of the system and, moreover, it is a factorized pure state.
\begin{theorem}  \label{thm: purr}
Assume that (D)-(H1)-(H2) hold true.
Let $|\psi_* \rangle$ be a unit-norm vector in $\mathbb{C}^{N}$, and assume that there exist a nontrivial interval $I$ and a curve $\brho: I \to \mathcal{P}$, solution of \eqref{eq: GKSL}, such that 
that $\brho(t)\in \mathcal{W}_{\psi_*}$ for every $t\in I$.
Then there exists a unit-norm vector $|\phi\rangle\in \mathbb{C}^2$ such that 
\[
\brho(t)=|\phi\rangle \langle \phi|\otimes |\psi_*\rangle \langle \psi_*| \qquad \forall t\in I,
\]
i.e., $\brho$ is constant and a pure factorized state of the composite system.
\end{theorem}
\begin{proof}
Let $\overline{\boldsymbol{x^{\rho}}}(t)$ be the coherence vector representation of $\brho(t)$. Thanks to Lemma~\ref{lem: fact},  there exist two curves $\bx^A$, $\boldsymbol{w}^A : I \to \mathbb{R}^3$ such that
\begin{gather}
\overline{\boldsymbol{x^{\rho}}}(t)=(1/\sqrt{2N},\bx^A(t),\sqrt{2N}\bx^A(t)\otimes \bx^B_*,\bx^B_*) \label{eq: traj}\\
\dot{\overline{\boldsymbol{x^{\rho}}}}(t)=(0,\boldsymbol{w}^A(t),\sqrt{2N}\boldsymbol{w}^A(t)\otimes \bx^B_*,0,\ldots,0) \label{eq: traj vel}
\end{gather}
for every $t\in I$. 

In order to compute the r.h.s. of  \eqref{eq: coherent controlled GKSL}, we take into account the expressions for $C_j$ and $D_j$ ( \eqref{eq: fijk} and \eqref{eq: dijk} in Appendix~\ref{sec: su(N)}, respectively), and we obtain the following equalities:
\begin{align}
(\widehat{H}_B\bx^B_*)_k&=
\begin{cases}
-\sqrt{\frac{N}{N-1}}\frulli_{k-1}\bb_* & \mbox{ if } \begin{cases}k=\beta_{N^2-1\, m}\\ m<N\end{cases} \\
\sqrt{\frac{N}{N-1}}\frulli_{k}\bb_* & \mbox{ if } \begin{cases}k=\alpha_{N^2-1\, m}\\ m<N\end{cases}\\
0 & \mbox{ otherwise}.
\end{cases}\\
(C_j^T\bx^B_*)_k&= \label{eq: cbstar}
\begin{cases}
\sqrt{\frac{N}{N-1}}\bb_* & \mbox{ if } \begin{cases}k=\beta_{N^2-1 m}\\
m<N,j=k-1\end{cases} \\
-\sqrt{\frac{N}{N-1}}\bb_* & \mbox{ if } \begin{cases}k=\alpha_{N^2-1 m}\\m<N,j=k+1\end{cases}\\
0 & \mbox{ otherwise}.
\end{cases}\\
(D_j^T\bx^B_*)_k&=
\begin{cases}
\frac{2-N}{\sqrt{N(N-1)}}\bb_* & \mbox{ if } \begin{cases}j=k=\alpha_{nm} ,\beta_{nm}\\
m<n
\end{cases}\\
\frac{2}{\sqrt{N(N-1)}}\bb_* & \mbox{ if } \begin{cases}j=k=\gamma_{m} \\
m<N
\end{cases}\\
\frac{2(2-N)}{\sqrt{N(N-1)}}\bb_* & \mbox{ if } j=k=N^2-1\\
0 & \mbox{ otherwise}
\end{cases}
\end{align}
(for the definition of the coefficients $\alpha_{nm},\beta_{nm}$ and $\gamma_{n}$ we refer to Appendix~\ref{sec: su(N)}).

Substituting \eqref{eq: traj}  into \eqref{eq: coherent controlled GKSL}, we obtain
\begin{equation} \label{eq: vbdot}
\dot{\bx}^B(t)=\sum_{j=1}^{N^2-1} \big(2\sqrt{N}(h_{1j},h_{2j},h_{3j})\cdot \bx^A(t)+\frulli_j\big)C_j^T\bx^B_*. 
\end{equation}
Imposing \eqref{eq: traj vel} and taking into account 
\eqref{eq: cbstar}, we obtain the following family of constraints
\begin{equation} \label{eq: ortogolli}
2\sqrt{N}(h_{1j},h_{2j},h_{3j})\cdot \bx^A(t)+\frulli_j=0,
\end{equation}
holding true for every $ 1\leq j <N^2-1$ such that $j\neq n^2-1, n\leq N$.

On the other hand, the equation for $\dot{\bx}^{AS}$ reads
\begin{align}
\dot{\bx}^{AS}(t)&=
2\sqrt{\frac{2}{N}} \sum_{l=1}^3 T_l\bx^A(t)\otimes h_{l\blacktriangle}\\
&+\sqrt{2N} (\Gamma+\widehat{H}_A)\bx^A(t)\otimes \bx^B_*+\sqrt{2N} \bx^A(t)\otimes\widehat{H}_B \bx^B_*\\
&+ 2\sqrt{N}\sum_{l=1}^3\sum_{j=1}^{N^2-1} h_{lj} T_l\bx^A(t)\otimes D_j^T\bx^B_*\\
&+\sqrt{2}\sum_{j=1}^{N^2-1}\left(\begin{smallmatrix}
                                     h_{1j}\\ h_{2j}\\ h_{3j}
                                    \end{smallmatrix}
\right)\otimes C_j^T\bx^B_*+\kappa \otimes \bx^B_* , \label{eq: vabdotnn}
\end{align}
where 
we set $h_{l\blacktriangle}=(h_{l1},\ldots,h_{l N^2-1})$, $l=1,2,3$.

Equation \eqref{eq: traj vel} imposes that $\dot{\bx}^{AS}_k(t)= 0$ whenever $k\neq (N^2-1)l$, $l=1,2,3$. 
Consider then an integer $1\leq k\leq 3(N^2-1)$ such that $k\notin\{ (N^2-1)l: l=1,2,3\}$, and call $a_k$ and $b_k$ the unique coefficients $1\leq a_k\leq 3$ and $1\leq b_k<N^2-1$ such that $k=(N^2-1)(a_k-1)+b_k$.
It can be verified by explicit (simple but tedious) computations that, if $b_k=r^2-1$ for some $1\leq r<N$, then the $k$-th  component of the right-hand side of \eqref{eq: vabdotnn} is zero.
Let us then consider $k$ such that $b_k\neq r^2-1$, and assume that $b_k=\alpha_{nm}$ for some $1\leq m<n\leq N$.
Then 
\begin{align}\left(
\begin{smallmatrix}\dot{\bx}^{AS}_k\\\dot{\bx}^{AS}_{k+1}\end{smallmatrix}\right)&= 
2\sqrt{\frac{2}{N}} \left(
\begin{smallmatrix}\big(h_{\bullet \alpha_{nm}}\wedge \bx^A(t)\big)_{a_k}\\\big(h_{\bullet \beta_{nm}}\wedge \bx^A(t)\big)_{a_k}\end{smallmatrix}\right) \\
&+\sqrt{2N}\sqrt{\frac{N}{N-1}}b_*\left(
\begin{smallmatrix}
\frulli_{\beta_{nm}}\\-\frulli_{\alpha_{nm}}
\end{smallmatrix}
\right)\bx^A_{a_k}(t)\\
&+ 2\sqrt{N} \frac{2-N}{\sqrt{N(N-1)}}b_*\left(
\begin{smallmatrix}\big(h_{\bullet \alpha_{nm}}\wedge \bx^A(t)\big)_{a_k}\\\big(h_{\bullet \beta_{nm}}\wedge \bx^A(t)\big)_{a_k}\end{smallmatrix}\right) \\
&+\sqrt{2}\sqrt{\frac{N}{N-1}}b_*\left(
\begin{smallmatrix}
h_{a_k \beta_{nm}}\\-h_{a_k\alpha_{nm}} 
\end{smallmatrix}
\right)\\
&=
\left(
\begin{smallmatrix}
 \sqrt{2N} \big(h_{\bullet \alpha_{nm}}\wedge \bx^A(t)\big)_{a_k}-h_{a_k\beta_{nm}} -\sqrt{N} \frulli_{\beta_{nm}}\bx^A_{a_k}(t) \\
  \sqrt{2N} \big(h_{\bullet \beta_{nm}}\wedge \bx^A(t)\big)_{a_k}+h_{a_k\alpha_{nm}} +\sqrt{N} \omega_{\alpha_{nm}}\bx^A_{a_k}(t) 
\end{smallmatrix}
\right)
\end{align}
where for every $1\leq j\leq N^2-1$ we set
\[h_{\bullet j}=(h_{1j},h_{2j},h_{3j}).\]

Together with \eqref{eq: ortogolli}, the equation above provides the following constraints on $\bx^A(t)$, holding true for every $1\leq m <n<N$:
\begin{equation} \label{eq: vincoletti}
\begin{cases}
h_{\bullet\alpha_{nm}} \cdot\bx^A(t)=-\frac{\frulli_{\alpha_{nm}}}{2\sqrt{N}}\\
 h_{\bullet\beta_{nm}}\cdot\bx^A(t)=-\frac{\frulli_{\beta_{nm}}}{2\sqrt{N}}\\
h_{\bullet\alpha_{nm}}  \wedge \bx^A(t)-\frac{\frulli_{\beta_{nm}}}{\sqrt{2}}\bx^A(t)=\frac{1}{\sqrt{2N}}h_{\bullet\beta_{nm}} \\
 h_{\bullet\beta_{nm}} \wedge\bx^A(t)+\frac{\omega_{\alpha_{nm}}}{\sqrt{2}}\bx^A(t)=-
 \frac{1}{\sqrt{2N}}h_{\bullet\alpha_{nm}}. 
\end{cases}
\end{equation}

It is left to show that every 3-dimensional curve $\bx^A(t)$ satisfying \eqref{eq: vincoletti} and such that $\|\bx^A(t)\|^2\leq 1/2N$ for every $t$ must be constant with $\|\bx^A(t)\|^2= 1/2N$. To do so, we must distinguish 2 cases.

 $\boldsymbol{\frulli_{\alpha_{nm}}=\frulli_{\beta_{nm}}=0}$ \textbf{\textit{for every}} $\boldsymbol{m,n}$.  By (H2), there exists at least one pair $(m,n)$ such that at least one between $h_{\bullet \alpha_{nm}}$ and $h_{\bullet \beta_{nm}}$ is not the zero vector. Without loss of generality, assume that $h_{\bullet \alpha_{nm}}\neq 0$. The fourth equation in \eqref{eq: vincoletti} implies that $h_{\bullet \beta_{nm}}\neq 0$ and, moreover, that the three vectors $(h_{\bullet \alpha_{nm}},\bx^A(t),h_{\bullet \beta_{nm}})$ constitute a right-handed orthogonal triple for every $t\in I$. 
 Taking the norm of the last two equations in~\eqref{eq: vincoletti}, we see that, if there exists $\bx^A(t)$ satisfying \eqref{eq: vincoletti}, then $\|\bx^A(t)\|\equiv\frac{1}{\sqrt{2N}}$ for every $t\in I$.

\textbf{\textit{There exists a pair $\boldsymbol{(n,m)}$ such that $\boldsymbol{\frulli_{\alpha_{nm}}\neq 0}$ or $\boldsymbol{\frulli_{\beta_{nm}}\neq 0}$}}.
Without loss of generality, assume that $\frulli_{\alpha_{nm}}\neq 0$, which immediately imposes that both $\bx_A(t)$, $t\in I$, and  $h_{\bullet \alpha_{nm}}$ are not the zero vector (by first equation in \eqref{eq: vincoletti}).
Taking the scalar product with $\bx^A(t)$ of each member of the fourth equation  and comparing with the first equation, we see that the system is satisfied only if $\|\bx^A(t)\|\equiv\frac{1}{\sqrt{2N}}$.
If $h_{\bullet \beta_{nm}}$ and $h_{\bullet \alpha_{nm}}$ are linearly dependent, then $\bx^A$ is parallel to $h_{\bullet \alpha_{nm}}$ for every $t$ (otherwise the third equation would not be satisfied), thus constant.
Else, the projection of $\bx^A(t)$ onto the linear span of $\{h_{\bullet \alpha_{nm}},h_{\bullet \beta_{nm}}\}$ is constant thus, by the constancy of the norm, the claim holds true.
\end{proof}

\medskip
Thanks to Theorem~\ref{thm: purr}, we can extract other necessary conditions for  global stabilization. Indeed,
 substituting \eqref{eq: traj}  into \eqref{eq: coherent controlled GKSL}, and taking into account the fact that $\bx^A(t)$ is constant, we obtain that $\bx^A(t)\equiv \bx^A_*$, where $\bx^A_*$ 
satisfies
\begin{equation}\label{eq: vadot}
\big(\Gamma+\sum_{l=1}^3\eta_lT_l \big)\bx^A_*+\frac{\kappa}{\sqrt{2N}}=0, 
\end{equation}
where we set
\begin{equation} \label{eq: eta}
\eta_l=\alpha_l-\frac{4}{\sqrt{2}}\sqrt{1-\frac{1}{N}}h_{lN^2-1}, \quad l=1,2,3.
\end{equation}
It is worth remarking that, if $\bx^A(t)\equiv\bx^A_*$ and the couplings $h_{ij}$ are chosen according to \eqref{eq: ortogolli}, then
$\dot{\bx}^{AS}(t)\equiv 0$, as it should be. This has already been proved for the components $k\neq a(N^2-1)$, $a=1,2,3$, in the proof of Theorem~\ref{thm: purr}.
Let us now consider $k=a(N^2-1)$, for some $a=1,2,3$.

By antisymmetry, $
(C_j^T\bx^B_*)_k=0$ for every $j$. On the other hand, $(D_j^T\bx^B_*)_k\neq0$ only if $j=N^2-1$ and, if so, we have  $(D_j^T\bx^B_*)_k=-\frac{2(2-N)}{\sqrt{2}N}$.
Then 
\begin{align}
\dot{\bx}_k^{AS}(t) &=2\sqrt{\frac{2}{N}}\sum_{l=1}^3 h_{l N^2-1} (T_l\bx^A_*)_a \nonumber \\
&+\sqrt{2N}\big(\Gamma\bx^A_* +\sum_{l=1}^3\alpha_l T_l\bx^A_*\big)_a \bb_* \nonumber \\
&-2\sqrt{N}\frac{2(2-N)}{\sqrt{2}N}  \sum_{l=1}^3  h_{l N^2-1} (T_l\bx^A_*)_a +  \kappa_a\bb_* \nonumber \\
&=  \sqrt{2N}b_* \Bigg( \Gamma\bx^A_*+ \sum_{l=1}^3\Big(\alpha_l-\frac{4}{\sqrt{2}}\sqrt{1-\frac{1}{N}}\Big)T_l\bx^A_*\Bigg)_k  \nonumber \\
&+\bb_*\kappa_k. \label{eq: VAB ultimi}
\end{align}

In particular, if $\bx^A_*$ is a solution of \eqref{eq: vadot}, then the right-hand side of \eqref{eq: VAB ultimi} is zero for every $a=1,2,3$.

\subsection{Existence of a factorized pure equilibrium}

The existence of the factorized pure equilibrium prescribed by Theorem~\ref{thm: purr} relies on the existence of a solution of equation~\eqref{eq: vadot} with norm $1/\sqrt{2N}$.
The following Proposition shows that, provided that at least one noise operator is not normal, it is always possible to choose the coupling in such a way that such a vector exists.

\begin{lem} \label{lem: va esiste}
Let $m\geq 1$,  and consider $\Gamma=\sum_{i=1}^m \Gamma_i$ and $\kappa=\sum_{i=1}^m \kappa_i$, where $\Gamma_i$ and $\kappa_i$ are as in  \eqref{eq: gammai} and \eqref{eq: kappai}, respectively. Assume that $b_i\neq 0$ for at least one $1\leq i\leq m$.
Then there exist $\boldsymbol{\eta}\in \mathbb{R}^3$ such that $M_{\boldsymbol{\eta}}=\Gamma+\sum_{i=1}^3\eta_iT_i$ is invertible and $\|M_{\boldsymbol{\eta}}^{-1}\kappa\|=1$.
\end{lem}

\begin{proof}
Choosing
\begin{gather}
\eta_1=-\sum_{i=1}^3a_ib_i\sin(\chi_i) \label{eq: eta1}\\ 
\eta_2=\sum_{i=1}^3a_ib_i\cos(\chi_i)\label{eq: eta2},
\end{gather}
one has 
\[
M_{\boldsymbol{\eta}}=\left(\begin{smallmatrix}
 -\sum_i(2a_i^2+\frac{1}{2}b_i^2) & -\eta_3 & 2\sum_i a_ib_i\cos(\chi_i)\\
 \eta_3 &-\sum_i(2a_i^2+\frac{1}{2}b_i^2) & 2\sum_i a_ib_i\sin(\chi_i)\\
0 & 0 & -\sum_i b_i^2
 \end{smallmatrix}\right),\]
 which is invertible for every value of $\eta_3$.  The thesis comes from the fact that
$M_{\boldsymbol{\eta}} \left(\begin{smallmatrix}
                              0 \\ 0 \\ 1
                             \end{smallmatrix}
\right)=-\kappa$.
\end{proof}

As a consequence of Lemma~\ref{lem: va esiste}, we have the following.

 \begin{proposition} \label{prop: existence eq}
 Let $h_A$ and $h_S$ be a 2-dimensional and a $N$-dimensional traceless Hermitian matrices, respectively. 
 Let $\ell_k$, $k=1,2,3$, be three (possibly null) upper triangular 
 2-dimensional matrices, and assume that at least one of them is not normal.
 
Then, for every unit norm $|\psi_*\rangle \in \mathbb{C}^N$, there always exists a matrix $H_I$ of the form \eqref{eq: HI} such that $\left(\begin{smallmatrix}
    1 & 0 \\ 0 & 0                                                                                                                                                                                                                                                                               \end{smallmatrix}
\right)\otimes |\psi_*\rangle \langle\psi_*|$ is an equilibrium of the GKSL equation with Hamiltonian
\[
H=h_A\otimes \id_N + H_I +\id_2 \otimes h_S
\]
and noise operators $\ell_i\otimes \id_N$, $i=1,2,3$.
 \end{proposition}

\subsection{The importance of being unique}

The existence of a pure factorized equilibrium of GKSL equation contained in $\mathcal{W}_{\psi_*}$ is a necessary condition for the stabilizability, but it is not sufficient. Indeed, on the one hand, several equilibria of the equation may exist (recall that, by linearity, the set of equilibria of the GKSL equation is convex, so that the existence of two equilibria automatically implies the existence of non-pure equilibria). 
On the other hand, the equilibrium may not be asymptotically stable (that is, only stable, as the GKSL equation does not admit unstable equilibria).

The two issues turn out to be only one, thanks to \cite[Theorem~1]{SchirmerWang10}, which states that an equilibrium of the GKSL equation is (globally) asymptotically
stable if and only if it is its unique equilibrium.
 Then, under the hypotheses of Proposition~\ref{prop: existence eq}, and using the couplings identified in previous sections, the system is globally asymptotically stable only if the target point is the unique equilibrium of the system.
 
Due to the high dimension, explicit computations of the rank of the matrix $M_H+M_{\mathcal{D}}$ are unpractical, even for $N=2$ and assigning numerical values to some parameters. 
Another approach is relying on the criteria developed  in the papers  \cite{Baumgartner2008a,Baumgartner2008b} (see also the conditions for subspace invariance in \cite{TiVi08,TiVi09}).
 Consider then any GKSL equation of the form \eqref{eq: GKSL}. 
 Following \cite{Baumgartner2008b}, we give the following definitions:            
\begin{definition}  
Let $\proj$ be an orthogonal projector on $\mathbb{C}^{2N}$. 
We say that
\begin{itemize}
 \item $\proj\mathbb{C}^{2N}$ is a {\bf lazy} subspace if and only if
\begin{equation} \label{eq: pigrizia}
L_k \proj=\proj L_k\proj \ \forall k.
\end{equation}
 \item $\proj\mathbb{C}^{2N}$ is a {\bf collecting} subspace if and only if
it is lazy and moreover
 \begin{equation} \label{eq: collecting}
\proj\big(iH -\frac{1}{2}\sum_{k} L_k^{\dagger}L_k\big)(\id_{2N} -\proj)=0.
\end{equation}
 \end{itemize} 
\end{definition}

\begin{definition}
 We define {\bf enclosure} a subspace $\mathcal{K}\subseteq \mathbb{C}^{2N}$ with
the property that, for all solutions $\brho(t)$ of GKSL equation,  the quantity {\rm tr}$\big(\proj_{\mathcal{K}} \brho(t)\big)$ (which is
the expectation value of the orthogonal projector onto this subspace)
 is constant in time.
\end{definition}
Lemma~7 in \cite{Baumgartner2008b} provides a criterion for 
$\mathcal{K}$ to be an enclosure: this occurs
if and only if $\proj_{\mathcal{K}}$ commutes both with the Hamiltonian $H$ and all noise operators $L_k$.

%

Collecting spaces and enclosures are tightly related with the equilibria of \eqref{eq: GKSL}, as the following results show.
\begin{lem} [\cite{Baumgartner2008b}] \label{lem: eq=> coll or encl}  
If $\brho$ is a stationary state of \eqref{eq: GKSL} and $\proj$ is the projector onto the range of $\brho$, then $\proj \mathbb{C}^{2N}$ is a
collecting subspace or an enclosure. 
\end{lem}

\begin{theorem}[\cite{Baumgartner2008b}] \label{thm: Baumgartner}
Let $\mathcal{K}$ be a subspace which is
a minimal enclosure or a minimal collecting subspace, containing no smaller enclosure or
collecting subspace. Then there exists one and only one stationary state supported by $\mathcal{K}$. Its
density matrix has maximal rank, ${\rm rank}(\brho) =\dim \mathcal{K}$.
\end{theorem}

 Assume that we are given a GKLS equation satisfying hypotheses (D)-(H1)-(H2) and the hypotheses of Proposition~\ref{prop: existence eq}. Assume moreover that the matrix elements of $H_I$ satisfy equations~\eqref{eq: vincoletti} and \eqref{eq: eta}, where $\eta_1$ and $\eta_2$ are chosen according to  
\eqref{eq: eta1}-\eqref{eq: eta2}. Previous results guarantee that $\left(\begin{smallmatrix}
    1 & 0 \\ 0 & 0                                                                                                                                                                                                                                                                               \end{smallmatrix}
\right)\otimes |\psi_*\rangle \langle\psi_*|$ is an equilibrium of the system. To see if such equilibrium is asymptotically stable, we must verify that span of the vector
\[
\left(\begin{smallmatrix}1\\0
\end{smallmatrix}
\right)\otimes \psi
\]
 is a minimal enclosure or a minimal collecting subspace.

\section{Stabilization of a qubit} \label{sec: qubit}

This section focuses on the stabilization of a qubit ($N=2$). While Section~\ref{sec: any dim} established necessary conditions for arbitrary $N$, completing them into necessary and sufficient conditions requires a thorough identification of all equilibria of the system. For generic $N$, this task is algebraically intractable due to quadratic growth of the dimension of the coherence representation. For $N=2$, however, low dimension allows for an explicit analysis of the equilibria, enabling us to derive exact structural constraints that guarantee global asymptotic stabilization.

First, let us remark that $N^2-1=3$ and that the three matrices $S_j$ coincide with the Pauli matrices (up to a normalization). In particular,
\[
S_{\alpha_{21}}=S_1=\frac{\sigma_1}{\sqrt{2}},\quad
S_{\beta_{21}}=S_2=\frac{\sigma_2}{\sqrt{2}},\quad
S_{\gamma_{2}}=S_3=\frac{\sigma_3}{\sqrt{2}}.
\]
Then, for every $j$, $D_j$ is the null matrix, while $f_{ijk}=\sqrt{2}\epsilon_{ijk}$, $\epsilon_{ijk}$ denoting the totally anti-symmetric tensor (with non-zero entries equal to one).                          
To simplify notations and for coherence with those  adopted in \cite{C26}, in this section we set $\lambda_{ij}=2h_{ij}$, $i,j=1,2,3$ (in this way, the interaction Hamiltonian becomes $H_I=\sum_{i,j=1}^3 \frac{\lambda_{ij}}{2}\sigma_i\otimes \sigma_j$). 
We also set $\nu_j=\frulli_j\sqrt{2}$.
                           
With these notations, the constraints~\eqref{eq: vincoletti} become
\[
\begin{cases}
 (\lambda_{12},\lambda_{22},\lambda_{32})\cdot\bx^A_*=-\frac{\nu_2}{2}\\
 (\lambda_{11},\lambda_{21},\lambda_{31})\cdot\bx^A_*=-\frac{\nu_1}{2}\\
 (\lambda_{11},\lambda_{21},\lambda_{31})\wedge\bx^A_*-\nu_2\bx^A_*=
 \frac{1}{2}(\lambda_{12},\lambda_{22},\lambda_{32})\\
 (\lambda_{12},\lambda_{22},\lambda_{32})\wedge \bx^A_*+\nu_1\bx^A_*=-\frac{1}{2}(\lambda_{11},\lambda_{21},\lambda_{31}).
\end{cases}
\]
Also, the parameters $\eta_l$ defined in \eqref{eq: eta} assume the following form
\[
\eta_l=\alpha_l-\lambda_{l3}, \quad l=1,2,3.
\]
Together with \eqref{eq: eta1}-\eqref{eq: eta2}, this imposes the following choices for the couplings:
\begin{gather}
\lambda_{13}=\alpha_1+\sum_{i=1}^3 a_ib_i\sin(\chi_i) \label{eq: lam13}\\ 
\lambda_{23}=\alpha_2-\sum_{i=1}^3 a_ib_i\cos(\chi_i).\label{eq: lam23}
\end{gather}

Set now 
$\bx^A_*=-M_{\boldsymbol{\eta}}^{-1}\frac{\kappa}{2}=(0,0,1/2)$. We choose two unit vectors $\boldsymbol{\epsilon}_1,\boldsymbol{\epsilon}_3\in \mathbb{R}^3$ such that $(\boldsymbol{\epsilon}_1,2\bx^A_*,\boldsymbol{\epsilon}_3)$ is a right-handed orthonormal basis. Then, under (H2), the system here above is satisfied only if there exist two constants $p_1,p_2\in \mathbb{R}$, not both zero,
such that
\begin{equation} \label{eq: scelta lambda gen} 
\begin{cases}\lambda_{\bullet 1}=p_1\boldsymbol{\epsilon}_1-2\nu_1\bx^A_*+p_2\boldsymbol{\epsilon}_3\\ 
\lambda_{\bullet 2}=-p_2\boldsymbol{\epsilon}_1-2\nu_2\bx^A_*+p_1\boldsymbol{\epsilon}_3
\end{cases}
\end{equation}


Equations~\eqref{eq: lam13}-\eqref{eq: lam23}-\eqref{eq: scelta lambda gen} constitute additional
necessary condition for the asymptotic stabilization of the qubit to the target state $|\psi_*\rangle \langle \psi_*|$.
As remarked above, they are not sufficient, as one has to ensure that there exists a unit-norm $\phi\in\mathbb{C}^2$ such that $|\phi\rangle\langle\phi|\otimes|\psi_*\rangle \langle \psi_*|$ is the unique equilibrium of \eqref{eq: GKSL}. To investigate whether it is the case, additional information about the noise operators is necessary. First, we consider the case of a single noise operator, and we  distinguish the case of a generic (full rank) operator from the case of a rank-one degenerate operator, starting with the latter case.
Then, we briefly discuss the case of multiple noise operators.

\subsection{$L=\sigma_+ \otimes \id$ (Ladder raising operator)}  \label{sec: sigma+}                  
In the normal form \eqref{eq: jump triang}, we assume $a=0$ and, without loss of generality, $b=1$ and $\chi=0$. Such (2-dimensional) operator is denoted with $\sigma_+$ and   is known as {\it ladder raising operator} for qubits. 
The associated matrix $\Gamma$ and the vector $\cappa$ have the expressions
\[
\Gamma=\left(\begin{smallmatrix}
          -\frac{1}{2} & 0 & 0\\
          0 & -\frac{1}{2} & 0\\
          0 & 0 & -1
         \end{smallmatrix}\right)\quad \mbox{and} \quad
         \cappa=\left(\begin{smallmatrix}
                 0 \\ 0\\  1
                \end{smallmatrix}\right),\]
respectively. 
First, we notice that $\|M^{-1}_{\boldsymbol{\eta}}\cappa\|=1$ only if $\eta_1=\eta_2=0$, which imposes  following choices of the coupling constants
\begin{equation} \label{eq: lambda3 sigma}
\lambda_{13}=\alpha_1\quad
\lambda_{23}=\alpha_2,
\end{equation}
and provides
 $\bx_*^A=(0,0,\frac{1}{2}).$
In order to fulfill equations~\eqref{eq: scelta lambda gen},
we can choose any pair of the form
\[
\lambda_{\bullet 1}=(\sin(\zeta),\cos(\zeta),-\nu_1)\qquad 
\lambda_{\bullet 2}=(\cos(\zeta),-\sin(\zeta),-\nu_2),
\]
where $\zeta\in\mathbb{R}$.
                    
We now inspect under which conditions the state 
\[\brho_*= \left(\begin{smallmatrix}1 & 0 \\ 0 & 0\end{smallmatrix}\right)  \otimes  \left(\begin{smallmatrix}0 & 0 \\ 0 & 1\end{smallmatrix}\right)\] is the unique equilibrium of GKSL.
Accordingly to the criteria developed in  \cite{Baumgartner2008b} (and recalled in the previous section), we must find all collecting subspaces and enclosures.

First, we need to identify the lazy spaces associated with the dynamics. For the sake of readability, we provide the details of the computations in Appendix~\ref{app: lazy}.
The projectors whose range is a lazy subspaces are:
\begin{enumerate}
\renewcommand{\theenumi}{\alph{enumi}}
\item \label{laz psi psi} All 1-dimensional subspaces of $\mathbb{C}^4$ generated by non-trivial vectors of the form $\boldsymbol{\varphi}^{\dagger}=(\varphi_1^*,\varphi_2^*,0,0)$. The projectors onto these spaces are of the form $\proj
=|\boldsymbol{\varphi}\rangle \langle \boldsymbol{\varphi}|$
(with $\boldsymbol{\varphi}$ of unit norm). Among them, we remark the space $\bV_{\psi_*}=\{(0,z,0,0): z\in \mathbb{C}\}$, which supports the target state $\brho_*$.
\item The 2-dimensional subspace generated by the vectors
\[ \big(1,0,0,
0\big)\quad \mbox{ and }\quad  
\big(
0 , 1 , 0 ,0\big), \]
i.e. the image of the projector
$ \proj=
\left( \begin{smallmatrix}
\id_2 & 0 \\
0 & 0
 \end{smallmatrix}\right).$
 \item  All 2-dimensional subspaces generated by the vectors
\[ \Big(\cos(\theta),0,0,
e^{-i\phi}\sin(\theta)\Big)\quad \mbox{ and }\quad  
\big(
0 , 1 , 0 ,0\big), \]
i.e. the images of the projectors
  \begin{equation} 
 \proj=
\left( \begin{smallmatrix}
\cos(\theta)^2\quad & 0 & 0 & \cos(\theta)\sin(\theta)e^{-i\phi}\\
0 & 1 & 0 & 0\\
0 & 0 & 0& 0\\
\cos(\theta)\sin(\theta)e^{i\phi} & 0 & 0& \sin(\theta)^2
 \end{smallmatrix}\right),
\end{equation}
for $\theta\in(0,\pi/2)$ and $\phi\in[0,2\pi]$.
\item  All 2-dimensional subspaces generated by the vectors
\[ \Big(0,\cos(\theta),
e^{-i\phi}\sin(\theta),0\Big)\quad \mbox{ and }\quad  
\big(
1,0 , 0 ,0\big), \]
i.e. the images of the projectors
  \begin{equation} 
 \proj=
\left( \begin{smallmatrix}
1& 0 & 0&  0\\
0 & \cos(\theta)^2\quad   & \cos(\theta)\sin(\theta)e^{-i\phi}& 0\\
 0& 
\cos(\theta)\sin(\theta)e^{i\phi} & \sin(\theta)^2 & 0 \\
0 & 0 & 0 & 0
 \end{smallmatrix}\right),
\end{equation}
for $\theta\in(0,\pi/2)$ and $\phi\in[0,2\pi]$.
 \item  All 2-dimensional subspaces generated by the vectors
\[ \big(\cos(\theta),
e^{-i\phi}\sin(\theta),0,0\big)\ \mbox{and} \
\big(
0 , 0 , \cos(\theta),
e^{-i\phi}\sin(\theta)\big), \]
i.e. the images of the projectors
  \begin{equation} \label{eq: P id psi psi}
 \proj=\id_2\otimes
\left( \begin{smallmatrix}
\cos(\theta)^2\quad  & \cos(\theta)\sin(\theta)e^{-i\phi}\\
\cos(\theta)\sin(\theta)e^{i\phi} &  \sin(\theta)^2
 \end{smallmatrix}\right),
\end{equation}
for $\theta\in[0,\pi/2]$ and $\phi\in[0,2\pi]$.
\item All 3-dimensional subspaces 
generated by the vectors
\begin{gather} 
(1,0,0,0),\quad  
\big(
0 , 1 , 0 ,0\big)\quad
 \big(0,0,\cos(\theta),
e^{-i\phi}\sin(\theta)\big).
\end{gather}
i.e. the images of the projectors
  \begin{equation} \label{eq: P 3 dim}
 \proj=
\left( \begin{smallmatrix}
1 & 0 & 0 & 0\\
0 & 1 & 0 & 0\\
0 & 0 & \cos(\theta)^2\quad & \cos(\theta)\sin(\theta)e^{-i\phi} \\
0& 0 &\cos(\theta)\sin(\theta)e^{i\phi} &  \sin(\theta)^2 
 \end{smallmatrix}\right),
\end{equation}
for $\theta\in[0,\pi/2]$ and $\phi\in[0,2\pi]$.
\end{enumerate}


We now verify which ones of the lazy subspaces above identified are also collecting. In order to do so, for every orthogonal  projection $\mathbb{P}$ on the candidate lazy subspace, we compute the matrix
$K=\proj\big(iH -\frac{1}{2}\sigma_-\sigma_+\otimes\id_2\big)(\id_4-\proj)$, which must be 
 the null matrix by equation \eqref{eq: collecting}. For each of the lazy subspaces identified above, we obtain the following results.
\begin{enumerate}
\renewcommand{\theenumi}{\alph{enumi}}
\item 
Explicit computations show that, if we substitute equations   \eqref{eq: lambda3 sigma} into the expression of $K$, the latter can be zero only when $\boldsymbol{\varphi}=(0,e^{-i\phi},0,0)$, for some $\phi$.
In this case, $\proj\mathbb{C}^{4}=\bV_{\psi_*}$  is a collecting subspace.
\item As the element $K_{14}$ of is always equal to $e^{i\zeta}$, 
$K$ is never the null matrix. So $\proj\mathbb{C}^{4}$ is not a collecting subspace.  

\item \label{coll P psi 1} 
Substituting equations~\eqref{eq: lambda3 sigma} into the expression of $K$, we obtain
\begin{align}
K_{43}&=\sin(\theta)\big(e^{i\phi}\cos(\theta)(i\alpha_1+\alpha_{2})+\sin(\theta)(\nu_1-i\nu_{2})\big)\\
K_{44}&=\frac{1}{8} \sin(2\theta) \Big(4 \cos(\phi+\zeta) +    4 i\cos(2\theta) \sin(\phi+\zeta) \\
   &- 
   i \sin(2\theta) (-i + 
   2 \alpha_3 + 2 \nu_3)\Big).
\end{align}
All other matrix elements of $K$ are proportional to $K_{34}$ or to $K_{44}$.
If $\alpha_1^2+\alpha_2^2\neq 0$ and $\nu_1^2+\nu_2^2\neq 0$, there always exists a unique pair $(\tilde{\theta},\tilde{\phi})$ for which $K_{43}=0$. 
Plugging these values into $K_{44}$, we cannot in principle exclude that $K_{44}=0$. 

However, we can prevent this to happen by using the control $u_3$. Indeed, recall that we can always  detune $\alpha_i$ with an action of the control in such a way that, for $(\theta,\phi)=(\tilde{\theta},\tilde{\phi})$,
we have $K_{44}\neq 0$, so that we destroy the collecting property of the subspace.
\item As $K_{14}=e^{i\phi}$, $\mathbb{PC}^4$ is not a collective subspace.

\item Consider now  projectors of the form \eqref{eq: P id psi psi}.
For $\theta=0$, we have that $K_{14}=e^{i\phi}$ while, for 
$\theta=\pi/2$, we have that $K_{41}=e^{i\phi}$, that is, 
$\mathbb{PC}^4$ is not collecting.

The case in which $\sin(2\theta)\neq0$ requires a longer analysis. Indeed, applying \eqref{eq: lambda3 sigma}, we obtain 
\begin{align*}
K_{31}&= \cos\theta(\sin \theta)^2 ( \cos\theta(i\alpha_1-\alpha_2)-e^{-i(\phi+\zeta)}\sin\theta)\\
K_{14}&= (\cos \theta )^3 (\cos \theta  e^{-i \zeta }-e^{-i \phi }(\alpha_2+i \alpha_1) \sin \theta )
\end{align*}
$K_{31}=0$  only if $\sin \theta = i (\alpha_1+i \alpha_2) \cos \theta  e^{i (\zeta+\phi )}$. Substituting this equality into the expression for $K_{14}$, we obtain
$K_{14}=e^{i\zeta}(\cos\theta)^2(1+\alpha_1^2+\alpha_2^2)$, which is never zero. 
Again, we can conclude that $\mathbb{PC}^4$ is not a collecting subspace.

\item Also for this case, we must treat separately the cases corresponding to different values of $\theta$. 

If $\sin\theta=0$, it is easy to verify that $K_{14}=e^{i\zeta}$. 
If instead $\cos\theta=0$, then $K$ is the null matrix only if $\nu_1=\nu_2=\alpha_1=\alpha_2=0$. We recall that, while the values of $\nu_i$ are imposed by the system, the values of $\alpha_i$ can possibly be modified by applying a constant control. Then we can always assume that at least one between $\alpha_1$ and $\alpha_2$ is not zero.

Assume now that both $\sin(\theta)\neq0$ and $\cos(\theta)\neq 0$, and plug \eqref{eq: lambda3 sigma} into the expression of $K$. In particular, we obtain
\begin{align}
K_{34}&= e^{-i \phi } (\cos \theta )^2 \Big(i\cos(2\theta)\big(\nu_1\cos\phi +\nu_2\sin\phi\big)\\
&+\nu_2\cos\phi-\nu_1\sin\phi+i\cos\theta\sin\theta\big(\lambda_{33}-\nu_3\big)
\Big).
\end{align}
If $\nu_1=\nu_2=0$, then, choosing  $\lambda_{33}\neq \nu_3$, we have $K_{34}\neq 0$. On the other hand, if at least one between $\nu_1$ and $\nu_2$ is nonzero, than there exist two unique  $\hat{\phi}_1,\hat{\phi}_2\in[0,2\pi]$ (with $\hat{\phi}_2-\hat{\phi}_1=\pi$) such that $\nu_2\cos\hat{\phi}_i-\nu_1\sin\hat{\phi}_i=0$. Associated with them, there exist $\hat{\theta}_1$ and $\hat{\theta}_2$ such that $\cos(2\hat{\theta}_j)(\nu_1\cos\hat{\phi}_j+\nu_2\sin\hat{\phi}_j)+\cos\hat{\theta}_j\sin\hat{\theta}_j(\lambda_{33}-\nu_3)=0$, $j=1,2$. As
\[
K_{14}=\cos\hat{\theta}_i\big(e^{i\zeta}\cos\hat{\theta}_i-e^{-i\hat{\phi}_i}\sin\hat{\theta}_i(i\alpha_1+\alpha_2)\big),
\]
and appropriate choice of $\alpha_1$ and $\alpha_2$ makes $K_{14}\neq 0$. 

Thus, with appropriate choices of the couplings and of the detuning,  $\mathbb{PC}^4$ is not a collecting subspace.
\end{enumerate}
Summing up,
the choices guaranteeing that the target state $\brho_*$ is the unique equilibrium of the system are:
\begin{itemize}
 \item if $\nu_1=\nu_2=0$, choose $\lambda_{33}\neq \nu_3$ and possibly modify $(\alpha_1,\alpha_2)$, in such a way that they are not both zero.
 \item if $\nu_1^2+\nu_2^2\neq 0$, then choose any $\lambda_{33}$ and apply a constant control to modify the values of $\boldsymbol{\alpha}$. In particular, first apply $u_3$ as prescribed at point c); then apply $(u_1,u_2)$ as prescribed in point f).
\end{itemize}


We summarize the results in the following statement.

\begin{theorem} \label{thm: un solo eq per sigmam}
Consider the quantum system \eqref{eq: GKSL} with only one noise operator $L=\sigma_+\otimes \id$ and $H$ as in Section~\ref{sec: preliminaries}. 
Then, possibly detuning the parameters $\alpha_i$ by means of a constant control, it is always possible to choose the coupling  parameters $\lambda$ in such a way that
$\brho_*$
is a globally asymptotically stable equilibrium of the system.
\end{theorem}

 \subsection{Generic noise operators} 
 \label{sec: generic channels}

Now we consider again one noise operator, and we assume that the parameters $a$ and $b$ in the normal form \eqref{eq: jump triang} are both nonzero.
We remark that assuming $b\neq 0$ is crucial: indeed, if not, the operator $\ell$ would be normal.

Under the assumptions made above, asymptotic stabilization towards the target set is always possible, as the following result states.

\begin{theorem} \label{thm: un solo eq per triang}
Let $\ell$ be as in \eqref{eq: jump triang} and assume that $ab \neq 0$. 
Then, possibly detuning the parameters $\alpha_i$ by means of a constant control, there exist a choice of the couplings $\lambda_{\bullet i}$, $i=1,2,3$, and a pure 2-dimensional density matrix $\rho_A$  on $\mathcal{H}_A$ such that 
$ 
    \rho_A\otimes \left(\begin{smallmatrix}
                0 &0 \\ 0 & 1
               \end{smallmatrix}\right)
$  is the only equilibrium of \eqref{eq: GKSL}.
\end{theorem}

\begin{proof}
Thanks to Lemma~\ref{lem: va esiste},  choosing   $\lambda_{1 3}$ and $\lambda_{23}$ according with \eqref{eq: eta1}-\eqref{eq: eta2} 
we obtain that $\bx^A_*=(0,0,1/2)$ satisfies equation~\eqref{eq: vadot}, independently on the value of $\lambda_{33}$. Then, we choose
\begin{equation} \label{eq: lambda gen}
\lambda_{\bullet2}=(\cos\zeta,-\sin\zeta,-\nu_2) \qquad
\lambda_{\bullet1}=(\sin\zeta,\cos\zeta,-\nu_1),
\end{equation}
for some $\zeta\in[0,2\pi)$.
To find the lazy spaces associated with these dynamics, we remark that the operator
 $\ell\otimes \id$ is diagonalizable and has two degenerate eigenvalues: $a$ corresponding to the eigenspace $V_+=\{(x_1,x_2,0,0): x_1,x_2\in \mathbb{C}\}$, 
 %
 and $-a$, with eigenspace $V_-=\mathbb{C}\{v_3,v_4\}$, where $v_3=(be^{-i\chi},0,-2a,0)$ and $v_4=(0,be^{-i\chi},0,-2a)$.
The lazy spaces are indeed the invariant spaces of $\ell\otimes \id$. In particular, besides the trivial spaces, we have:
\begin{enumerate}
\renewcommand{\theenumi}{\alph{enumi}}
  \item Any  1-dimensional subspace of $V_+$.
    \item Any  1-dimensional subspace of $V_-$.
 \item The eigenspace $V_+$.
 \item The eigenspace $V_-$.
 \item All 2-dimensional space generated by two eigenvectors pertaining to different eigenvalues. More precisely, the family of eigenspaces 
 \[
 V_{\eta,\theta}=\mathbb{C}
 \left\{ 
 \left(
 \begin{smallmatrix}
  \cos(\eta)\\ \sin(\eta)\\0 \\0
 \end{smallmatrix}\right),
\left(
 \begin{smallmatrix}
  \cos(\xi)b e^{-i\chi}\\ \sin(\xi)b e^{-i\chi}\\-2a\cos(\xi) \\-2a\sin(\xi)
 \end{smallmatrix}
 \right)
 \right\},
 \]
with $\eta,\xi\in[0,2\pi)$, 
 \item The direct sum of each eigenspace with any eigenvector pertaining to a different eigenvalue. More precisely, the two families of eigenspaces
 \begin{gather}
 W^+_{\xi}=\mathbb{C}
 \left\{ 
 \left(
 \begin{smallmatrix}
  1\\ 0\\0 \\0
 \end{smallmatrix}\right),
  \left(
 \begin{smallmatrix}
  0\\ 1\\0 \\0
 \end{smallmatrix}\right),
\left(
 \begin{smallmatrix}
 0\\0\\ \cos(\xi) \\ \sin(\xi)
 \end{smallmatrix}
 \right)
 \right\},
 \\
 W^-_{\xi}=\mathbb{C}
 \left\{ 
  \left(
 \begin{smallmatrix}
  1\\ 0\\0 \\0
 \end{smallmatrix}\right),
  \left(
 \begin{smallmatrix}
  0\\ 1\\0 \\0
 \end{smallmatrix}\right),
\left(
 \begin{smallmatrix}
  \cos(\xi)b e^{-i\chi}\\ \sin(\xi)b e^{-i\chi}\\-2a\cos(\xi) \\-2a\sin(\xi)
 \end{smallmatrix}
 \right)
 \right\},
 \end{gather}
 for $\xi\in[0,\pi)$.
 \end{enumerate}

\smallskip

As in the previous case, to verify which lazy subspaces  are also collecting,  we compute the matrix
$K=\proj\big(iH -\frac{1}{2}\ell^{\dagger}\ell\otimes\id_2\big)(\id_4-\proj)$, for each projector on a lazy subspace. 

 \begin{enumerate}
\renewcommand{\theenumi}{\alph{enumi}}
\item Explicit computations show that the only lazy collecting 1-dimensional subspace contained in $V_+$ is $\mathcal{W}_{\psi_*}$.
\item Straight computation show that $\mathbb{C}v_4$ is not collecting. Consider now the projector over $\cos(\xi)v_3+\sin(\xi)v_4$, with $\xi\neq \pi/2,3\pi/2$. 
By computation, we can show that 
\begin{align*}
2aK_{11}+be^{i\chi} K_{13}&=\frac{b\cos\xi\big(
A\sin\xi+(B+C\lambda_{33})\cos\xi\big)}{(4 a^2+b^2)e^{i\chi}},
\end{align*}
with 
\begin{gather}
 A=4a^2e^{-i\zeta}+2abe^{i\chi}(\nu_2-i\nu_1),\\ B=4 (\alpha_2-i \alpha_1) a^2+(\alpha _2+i \alpha_1) b^2 e^{2 i \chi }\\
 C=2iabe^{i\chi}.
\end{gather}
Then, there may exists $\xi$ such that $K$ is the null matrix only if the ratio $\frac{B+C\lambda_{33}}{A}$ is real.   
If $C/A\notin \mathbb{R}$, it is always possible to 
choose $\lambda_{33}$ in such a way that $\frac{B+C\lambda_{33}}{A}\notin \mathbb{R}$. Else, possibly slightly modifying $\alpha_1$ and/or $\alpha_2$, we obtain that $B/A\notin \mathbb{R}$, so that no value of $\lambda_{33}$ can make the term real. 
Thus, it is always possible to choose $\lambda_{33}$ and eventually modify $\boldsymbol{\alpha}$ in such a way that no 2-dimensional subspace of $V_-$ is collecting.
\item When considering the projector over $V_+$, we obtain $K_{14}=e^{-i\zeta}$. 
\item When considering the projector over $V_-$, we obtain
\[
K_{44}=-\frac{2 a^2 b^2 \left(4 a^2+2 i (\alpha_3-\lambda_{33})+b^2\right)}{\left(4 a^2+b^2\right)^2},
\]
which is not zero whenever $\lambda_{33}\neq \alpha_3$.
\item For every choice of $\eta,\xi\in[0,2\pi)$,
$V_{\eta,\xi}$ is not a collecting subspace. 
To prove this, we first notice that, for $(\xi,\eta)\in \{(\pi/2+k\pi,q\pi)\}$, $K_{14}=\frac{b^2e^{i\zeta}}{4a^2+b^2}$. 
For every $r>0$, call $B^r_{k,q}$ the ball of radius $r$ centered at $(\pi/2+k\pi,q\pi)$. 
We can thus fix $\epsilon>0$ such that, for every $k,q$ and every $(\xi,\eta)\in B^{\epsilon}_{k,q}$,  we have $|K_{14}|\geq \frac{b^2}{2(4a^2+b^2)}$, i.e. the matrix $K$ is not null. 

Then, by computations we can see that $K_{11}=f+g\lambda_3$, where 
\begin{align}
|g(\xi,\eta)|&=\frac{16(\sin\eta)^2}{8a^2+b^2/a^2(1-\cos(2(\eta-\xi))}\\
&\times \big(4a^2(\cos\eta)^2+b^2(\cos\xi)^2(\sin(\eta-\xi))^2\big)
\end{align}
and $f$ is a bounded function of $\xi,\eta$, and depending also on the parameters $a,b,\chi,\alpha_i,\nu_i$.
In particular, 
$|g(\xi,\eta)|$
 is uniformly bounded from below when $\xi$ and $\eta$ do not belong to any of the $B^{\epsilon}_{k,q}$. Then, it is possible to choose $\lambda_{33}$ large enough to assure that, for such values of $(\xi,\eta)$, $K_{11}\neq 0$. 
 \item First consider $W_{\xi}^+$. As $K_{13}=\sin \xi \left(\sin \xi  \left(i \alpha_1+\alpha_2 -a b e^{-i \chi }\right)-e^{i \zeta } \cos \xi \right)$, choosing any pair $(\alpha_1,\alpha_2)$ such that 
 $e^{-i\zeta}\left(i \alpha_1+\alpha_2 -a b e^{-i \chi }\right)\notin \mathbb{R}$ assures that $K_{13}$ is never null, for $\xi\neq k\pi$. On the other hand,
 if $\xi=k\pi$, then $K_{14}=e^{i\zeta}$. Then $W_{\xi}^+$ is not collecting. 
As for the space $W_{\xi}^-$, it is never collecting, since $K_{14}=-e^{-i\zeta}$.
 \end{enumerate}
Summing up, if we choose $\lambda_{33}\neq \alpha_3$ large enough and we possibly modify $\alpha_1$ and $\alpha_2$, we obtain that the only lazy collecting subspace is $\mathcal{W}_{\psi_*}$.
Thus, that GKSL has only one equilibrium point.
 \end{proof}

\subsection{Lindbladians containing several noise operators} 
\label{sec: several jumps}
Consider now the case in which the dissipative part of the GKSL equation is described by several noise operators, all satisfying (D). In \cite{C26}, we proved that it is possible to make the target state globally asymptotically stable only if all the jump operator are simultaneously triangularizable. The following result proves also a sufficient condition for this property to occur.

\begin{proposition} \label{prop: un solo eq per molti triang}
Assume that the noise operators can be simultaneously triangularized and that $\cappa\neq 0$ and $\det \Gamma\neq0$.
Then, possibly modifying the parameters $\alpha_i$ by means of a constant control, there exist a choice of the couplings $\lambda_{\bullet i}$ and a pure 2-dimensional density matrix $\rho_A$  on $\mathcal{H}_A$ such that 
$ 
    \rho_A\otimes \left(\begin{smallmatrix}
                1 &0 \\ 0 & 0
               \end{smallmatrix}\right)
$  is the only equilibrium of \eqref{eq: GKSL}.
\end{proposition}

\begin{proof}
We carry on the proof for two noise operators $\ell_1$ and $\ell_2$, the generalization to three noise operators being straightforward.

Assume that both $\ell_1$ and $\ell_2$ are in the normal form \eqref{eq: jump triang}. 
If both $\ell_1$ and $\ell_2$ are singular, then they are proportional to $\sigma_+$, the results of Section~\ref{sec: sigma+} hold. 
If not, assume, without loss of generality, that 
$\ell_1=\left(\begin{smallmatrix} a_1 & b_1e^{-i\chi_1}\\0 & -a_1\end{smallmatrix}
\right)$, with $a_1b_1\neq 0$. 

The lazy subspaces associated with the dynamics must satisfy  \eqref{eq: pigrizia} for both $L_1$ and $L_2$. Then, we just check which of the projectors onto lazy subspaces for $L_1$ also satisfy  \eqref{eq: pigrizia} for $L_2$. 
Straight computations show that, with the possible use of the control to slightly modify $\alpha_1$ or $\alpha_2$, $\mathcal{W}_{\psi_*}$ is the only lazy subspace of the dynamics.

Then GKLS has only one equilibrium point.
\end{proof}

%
%
%
%
%

\section{Examples}  \label{sec: example}

We illustrate the results by a numerical example.  We consider a system of interest composed by two spins {\bf S$_1$} and {\bf S$_2$}, evolving with the Hamiltonian
\[
h_{S}=\frac{5}{2\sqrt{2}}\sigma_z\otimes \id_2+\frac{1}{2\sqrt{2}}\id_2\otimes \sigma_z.
\]
Our aim is to stabilize it to the pure Bell state
\[
\psi_*=\frac{1}{\sqrt{2}}\Big( \left(\begin{smallmatrix}
1\\0                                                                           \end{smallmatrix}
\right) \otimes \left(\begin{smallmatrix}
0\\1                                                                           \end{smallmatrix}
\right) + \left(\begin{smallmatrix}
0\\1                                                                           \end{smallmatrix}
\right)\otimes \left(\begin{smallmatrix}
1\\0                                                                           \end{smallmatrix}
\right) \Big).
\]
Following the indirect stabilization scheme developed in this paper, we couple the system {\bf S}$=${\bf S$_1$+S$_2$} with a dissipative ancilla that evolves with the Hamiltonian $h_A=\frac{1}{2}\left(\begin{smallmatrix}2& 1-i\\ 1+i&-2                                                                                                                                                                                                                                                                \end{smallmatrix}
\right)$ and is subject to the dissipation $\ell=\sigma_+$.

First, we apply to the composite system the base change $\id_2\otimes O$, where
\[
O=\left(
\begin{smallmatrix}\
0 & \frac{1}{\sqrt{2}} & -\frac{1}{\sqrt{2}} &0\\
1 & 0 & 0 & 0\\
0 & 0 & 0 & 1\\
0 & \frac{1}{\sqrt{2}} & \frac{1}{\sqrt{2}} &0
\end{smallmatrix}
\right),
\]
so that $h_S$ becomes $h_S={\rm diag}\big(\frac{3}{\sqrt{2}},\frac{2}{\sqrt{2}},-\frac{2}{\sqrt{2}},-\frac{3}{\sqrt{2}}\big)=3 S_{6}+2S_9$.
Accordingly to the necessary conditions illustrated in Section~\ref{sec: any dim}, we choose the following couplings:
\begin{align}
 h_{\bullet 1}&=(0,2,0) &
  h_{\bullet 3}&=(0,0,0)\\
   h_{\bullet 4}&=(0,3,0) &
   h_{\bullet 6}&=(2,5,-3/\sqrt{2})\\
  h_{\bullet 8}&=(0,0,0)&
   h_{\bullet 9}&=(3,1,-2/\sqrt{2})\\
h_{\bullet 11}&=(5,0,0)&
h_{\bullet 13}&=(0,0,0)\\
h_{\bullet 15}&=(1/\sqrt{6},1/\sqrt{6},2/\sqrt{6}),
\end{align}
(the other ones can be recovered by means of \eqref{eq: vincoletti}, knowing  that $\bx_A^*=(0,0,1/2)$).
These choices are motivated by the fact that the resulting matrix $M_H+M_D$ possesses only one zero eigenvalue, all other having strictly negative real part. This ensures that the target state (of the composite system) is the only equilibrium. 


We randomly choose 10 initial point (we generate a random complex matrix $A$ of dimension 16, then we set $\brho(0)=\frac{A^{\dagger}A}{{\rm tr}(A^{\dagger}A)}$).
Figure~\ref{fig: simu1} displays the temporal evolution of the fidelity\footnote{We recall that the fidelity between two quantum states $\rho,\sigma$ is defined as $F(\rho,\sigma)=\big({\rm tr}\sqrt{\sqrt{\rho}\sigma\sqrt{\rho}}\big)^2$.} $F(t)$ across the ten resulting trajectories.

\begin{figure}[h!]
 \begin{center}
\includegraphics[width=7cm]{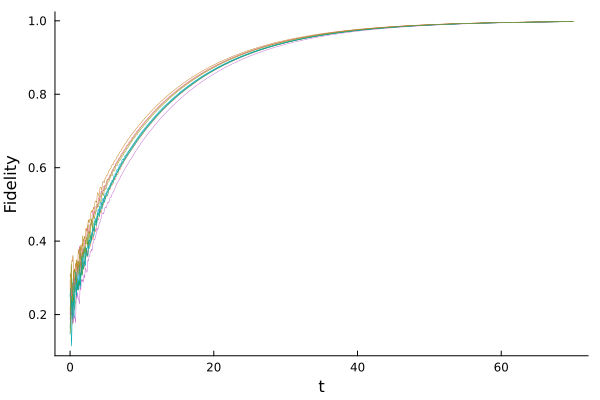}  
 \end{center}
 \caption{Fidelity $F({\rm tr}_A\brho(t),|\psi_*\rangle\langle \psi_*|)$.} \label{fig: simu1}
\end{figure}

 \section{Conclusions}

 This paper establishes some necessary conditions for the stabilization of a $N$-level quantum system coupled to a noisy controlled ancilla qubit, when the target state for the system of interest is chosen to be any pure state. Our analysis highlights the following structural constraint: there must exists a stable pure factorized equilibrium of the composite system, whose reduced state is exactly the target state of the system of interest. This requirement, already observed in the qubit case, persists in any dimension.
 
In the case in which the system of interest is a qubit, we deepen the analysis by providing additional necessary and sufficient conditions, thus giving some explicit choices for the couplings. 


\appendix

\subsection{Standard generators and structure constants of $\mathfrak{su}(N)$} \label{sec: su(N)}

Quite surprisingly, despite their massive use in particle physics, it is quite difficult to find explicit formulas for the standard generators and the structure constants of $\mathfrak{su}(N)$, if $N> 4$. The paper \cite{BossionHuo2021} seems the most complete reference available. For completeness' sake, we recall here below the main definitions and formulas that we are using throughout the paper (in order to work with orthonormal bases of the space of Hermitian matrices, we use a different normalization with respect to the standard one). 

For $N\geq2$, we consider the following  
$N^2-1$ matrices 
\begin{gather}
S_{\alpha_{nm}}=\frac{1}{\sqrt{2}}(|m\rangle \langle n|+|n\rangle \langle m|),\\
S_{\beta_{nm}}=-\frac{i}{\sqrt{2}}(|m\rangle \langle n|-|n\rangle \langle m|),\\
S_{\gamma_{n}}=\frac{1}{\sqrt{n(n-1)}}\Big(\sum_{l=1}^{n-1} |l\rangle \langle l|+(1-n)|n\rangle \langle n|\Big)
\end{gather}
where, for every $1\leq m<n\leq N$, $|n\rangle$ denotes the $n$-th element of the canonical basis of $\mathbb{R}^{N^2-1}$ and the coefficients $\alpha_{nm},\beta_{nm}$ and $\gamma_{n}$ are defined as
\begin{gather}
\alpha_{nm}=n^2+2(m-n)-1,\qquad 
\beta_{nm}=n^2+2(m-n),\\
\gamma_{n}=n^2-1.
\end{gather}
It is easy to verify that they constitute a orthonormal basis of the set of traceless Hermitian $N$-dimensional matrices.

For every $N$, the matrices $S_j$, $j=1,\ldots,N^2-1$ satisfy the following commutation and anticommutation relations:
\begin{gather}
[S_l,S_j]=i\sum_{k=1}^{N^2-1} f_{ljk} S_k\\
\{S_l,S_j\}=\frac{1}{N}\delta_{lj}\id_N+\sum_{k=1}^{N^2-1} d_{ljk} S_k,
\end{gather}
where the constants $f_{ljk}$ are totally antisymmetric and the constants $d_{ljk}$ are totally symmetric. 
Their expression can be found in \cite{BossionHuo2021} (up to a normalization factor). Thanks to a suitable choice of the coordinate basis, in this paper we only use the coefficients whose last index ($k$) is equal to $\gamma_N=N^2-1$. In this case, for every $1\leq m<N$, we have: 
\begin{equation} \label{eq: fijk}
f_{lj\gamma_N}=
\begin{cases}
 \sqrt{\frac{N}{N-1}} & \mbox{ if } l=\alpha_{Nm},j=\beta_{Nm}\\
  -\sqrt{\frac{N}{N-1}} & \mbox{ if } l=\beta_{Nm},j=\alpha_{Nm}\\
  0 & \mbox{otherwise}.
\end{cases}
\end{equation}
Concerning the symmetric constants $d_{lj\gamma_N}$, they are non-zero only if $l=j$ and, in particular, we have:
\begin{equation} \label{eq: dijk}
d_{ll \gamma_N}=\begin{cases}\frac{2}{\sqrt{N(N-1)}} & \mbox{ if } l=\gamma_m, m<N\\
\frac{2-N}{\sqrt{N(N-1)}} & \mbox{ if } \begin{array}{c}l=\alpha_{nm}\mbox{ or } \beta_{nm},\\
 m<n<N                                         \end{array}
\\
\frac{2(2-N)}{\sqrt{N(N-1)}} & \mbox{ if } l=\gamma_N.
\end{cases}
\end{equation}


For $j=1,\ldots,N$, we define the  square $(N^2-1)$-dimensional matrix $C_j$ and $D_j$, setting
\begin{equation} \label{eq: matrici C e D}
(C_j)_{km}=f_{jkm}\qquad
(D_j)_{km}=d_{jkm}.
\end{equation}

\subsection{Lazy subspaces associated with $\sigma_+\otimes \id$} \label{app: lazy}

Following \cite[Lemma~9]{Baumgartner2008b}, a subspace $\mathcal{K}$ of $\mathbb{C}^n$ is a lazy subspace if and only if its associated orthogonal projector satisfies equation~\eqref{eq: pigrizia}. 
To find lazy subspaces associated with $\sigma_+\otimes \id$,  we start by writing the generic  projector $\proj$ on $\mathbb{C}^4$ in the block form  $\proj=\left(\begin{smallmatrix} A & B\\ B^{\dagger} & D\end{smallmatrix}\right)$, where $A,B$ and $D$ are square of dimension 2 and both $A$ and $D$ are Hermitian. Remark that $\proj$ is an orthogonal projector if and only if
\begin{gather}
A^2+B B^{\dagger}=A \label{eq: proj 1} \\
AB+BD=B \label{eq: proj 2} \\
B^{\dagger}B+D^2=D. \label{eq: proj 3}
\end{gather}
To fix notations, we call $a_{ij}$ the matrix elements of $A$ and $d_{ij}$ the matrix elements of $D$, respectively.

For $\ell=\sigma_+$,  equation~\eqref{eq: pigrizia} reads
\begin{equation} \label{eq: pigrizia +}
\begin{pmatrix}
 B^{\dagger} & D \\
 0 & 0
\end{pmatrix}=
\begin{pmatrix}
 AB^{\dagger} & AD \\
 (B^{\dagger})^2 & B^{\dagger}D
\end{pmatrix}
\end{equation}
In particular, $B$ must be nilpotent, that is, either it is the null matrix or has one of the following forms
\[
\left(\begin{smallmatrix}
 \alpha & \beta \\
 -\frac{1}{\beta} & -\alpha
\end{smallmatrix}\right)
\quad 
\left(\begin{smallmatrix}
 0 & \beta \\
 0 & 0
\end{smallmatrix}\right)
\quad \left(\begin{smallmatrix}
 0 & 0 \\
 \beta & 0
\end{smallmatrix}\right),\qquad \alpha\in\{-1,1\},\ \beta\neq 0.
\]
The first option is ruled out by equation~\eqref{eq: proj 3}. 
Let us then assume that $B$ is a nontrivial upper triangular nilpotent matrix. Carrying on all the products, we can see that the matrix elements of $\mathbb{P}$ must satisfy the following constraints:
\[
\begin{cases}
d_{11}=d_{12}=a_{12}=0\\
a_{11}+d_{22}=1\\
a_{22}=1\\
d_{22}^2+|\beta|^2=d_{22}.
\end{cases}
\]
The second order equation for $d_{22}$ has real roots only for $|\beta|\leq 1/2$ and, in this case, the roots are non-negative. Thus we can set $a_{11}=\cos(\theta)^2$ and $d_{22}=\sin(\theta)^2$, for some $\theta\in\mathbb{R}$, which yields $\beta=\sin(\theta)\cos(\theta)e^{i\phi}$, for some $\phi\in\mathbb{R}$. 

For fixed $\theta,\phi$,
$\mathbb{P}$  is the projector on the 2-dimensional space spanned  by the vectors
\[ \Big(\cos(\theta),0,0,
e^{-i\phi}\sin(\theta)\Big)\quad \mbox{ and }\quad  
\big(
0 , 1 , 0 ,0\big). \]
%
 In order to avoid redundancy, we restrict the value of the parameters to $\theta\in(0,\pi/2)$ and $\phi\in[0,2\pi)$.

 If we assume that $B$ is a nontrivial lower triangular nilpotent matrix,  analogous computations  
 identify the 2-dimensional space spanned  by the vectors
\[ \Big(0,\cos(\theta),0,
e^{-i\phi}\sin(\theta),0\Big)\quad \mbox{ and }\quad  
\big(
1, 0 , 0 ,0\big) \]
as lazy subspace for the dynamics. As above, we restrict to $\theta\in(0,\pi/2)$ and $\phi\in[0,2\pi)$.

We finally consider the case in which $B$ is the null matrix. 
We have several cases.
\begin{enumerate}
\item If $D$ is the null matrix, then any choice of  $A$ such that $A^2=A$ provides a projector on a lazy subspace.
In particular, we have
\[
\mathbb{P}=\begin{pmatrix}  \id & 0 \\ 0 & 0\end{pmatrix}\quad \mbox{ and }
\mathbb{P}=\begin{pmatrix} \pur{\psi} & 0\\ 0 &0\end{pmatrix}, |\psi\rangle \in \mathbb{C}^2.
\]
In the first case, the lazy subspace has dimension 2. In the second, it has dimension 1.
 \item If $D$ is not the null matrix, exploiting \eqref{eq: proj 3} and \eqref{eq: pigrizia +}, we see that the following equation must hold
 \[
 (A-\id)D=0.
 \]
 Together with \eqref{eq: proj 1}, this imposes $A=\id$ or $A=D$.

\end{enumerate}

\bibliographystyle{plain}        
\bibliography{../biblio-qopen}           

\begin{thebibliography}{10}

\bibitem{alicki-fannes}
R.~Alicki and M.~Fannes.
\newblock {\em Quantum Dynamical Systems}.
\newblock Oxford University Press, 2001.

\bibitem{altafiniJMP}
C.~Altafini.
\newblock Controllability properties for finite dimensional quantum {M}arkovian
  master equations.
\newblock {\em J. Math. Phys.}, 44(6):2357--2372, 2003.

\bibitem{altaf-tic-12}
C.~Altafini and F.~Ticozzi.
\newblock Modeling and control of quantum systems: An introduction.
\newblock {\em IEEE Transactions on Automatic Control}, 57(8):1898--1917, 2012.

\bibitem{avron}
J.~Avron and O.~Kenneth.
\newblock An elementary introduction to the geometry of quantum states with
  pictures.
\newblock {\em Reviews in Mathematical Physics}, 32(02):2030001, 2020.

\bibitem{Baumgartner2008b}
B.~Baumgartner and H.~Narnhofer.
\newblock Analysis of quantum semigroups with {GKS}-{L}indblad generators:
  {II}. {G}eneral.
\newblock {\em Journal of Physics A: Mathematical and Theoretical},
  41(39):395303, sep 2008.

\bibitem{Baumgartner2008a}
B.~Baumgartner, H.~Narnhofer, and W.~Thirring.
\newblock Analysis of quantum semigroups with {GKS}--{L}indblad generators: I.
  simple generators.
\newblock {\em Journal of Physics A: Mathematical and Theoretical},
  41(6):065201, jan 2008.

\bibitem{BossionHuo2021}
D.~Bossion and P.~Huo.
\newblock General formulas of the structure constants in the $\mathfrak{su}(n)$
  {L}ie {A}lgebra.
\newblock {\em arXiv preprint arXiv:2108.07219}, 2021.

\bibitem{breuer-petr}
H.-P. Breuer and F.~Petruccione.
\newblock {\em The Theory of Open Quantum Systems}.
\newblock Oxford University Press, 2002.

\bibitem{C26}
F.~C. Chittaro.
\newblock Ancilla-assisted stabilization of a qubit.
\newblock {\em IFAC-PapersOnLine}, 2026.
\newblock In press. Presented at the 23rd IFAC World Congress, Busan, South
  Korea.

\bibitem{FroSchu16}
J.~Fröhlich and B.~Schubnel.
\newblock The preparation of states in quantum mechanics.
\newblock {\em Journal of Mathematical Physics}, 57(4):042101, 04 2016.

\bibitem{gorini-kos-sud}
A.~Gorini, A.~Kossakowski, and E.C.G. Sudarshan.
\newblock Completely positive dynamical semigroups of {N}-level systems.
\newblock {\em J. Math. Phys}, 17:821, 1976.

\bibitem{HaMuMu22}
P.M. Harrington, E.J. Mueller, and K.W. Murch.
\newblock Engineered dissipation for quantum information science.
\newblock {\em Nat Rev Phys}, 2022.

\bibitem{dirr-helmke}
I.~Kurniawan, G.~Dirr, and U.~Helmke.
\newblock Controllability aspects of quantum dynamics: a unified approach for
  closed and open systems.
\newblock {\em IEEE Transactions on Automatic Control}, 57(8):1984--1996, 2012.

\bibitem{lindblad}
G.~Lindblad.
\newblock On the generators of quantum dynamical semigroups.
\newblock {\em Comm Math Phys}, 48(2):119--130, 1976.

\bibitem{PhysRevLett.119.150502}
Y.~Lu, S.~Chakram, N.~Leung, N.~Earnest, R.~K. Naik, Z.~Huang, P.~Groszkowski,
  E.~Kapit, J.~Koch, and D.~I. Schuster.
\newblock Universal stabilization of a parametrically coupled qubit.
\newblock {\em Phys. Rev. Lett.}, 119:150502, Oct 2017.

\bibitem{MiLeVo13}
M.~Mirrahimi, Z.~Leghtas, and U.~Vool.
\newblock Quantum reservoir engineering and single qubit cooling.
\newblock {\em IFAC Proceedings Volumes}, 46(23):424--429, 2013.
\newblock 9th IFAC Symposium on Nonlinear Control Systems.

\bibitem{nielsen-chuang}
M.~A. Nielsen and I.~L. Chuang.
\newblock {\em Quantum Computation and Quantum Information}.
\newblock Cambridge University Press, 2000.

\bibitem{RoRoSeIHP}
R.~Robin, P.~Rouchon, and L.A. Sellem.
\newblock Convergence of bipartite open quantum systems stabilized by reservoir
  engineering.
\newblock {\em Ann. Henri Poincaré}, 26:1769--1819, 2025.

\bibitem{SchirmerWang10}
S.~G. Schirmer and X.~Wang.
\newblock Stabilizing open quantum systems by {M}arkovian reservoir
  engineering.
\newblock {\em Phys. Rev. A}, 81:062306, Jun 2010.

\bibitem{ShaHaLeal13}
S.~Shankar, M.~Hatridge, Z.~Leghtas, and f~et~al.
\newblock Autonomously stabilized entanglement between two superconducting
  quantum bits.
\newblock {\em Nature}, 504:419--422, 2013.

\bibitem{TiLuCapVi12}
F.~Ticozzi, R.~Lucchese, P.~Cappellaro, and L.~Viola.
\newblock Hamiltonian control of quantum dynamical semigroups: Stabilization
  and convergence speed.
\newblock {\em IEEE Transactions on Automatic Control}, 57(8):1931--1944, 2012.

\bibitem{TiSchiWa-TAC}
F.~Ticozzi, S.~G. Schirmer, and X.~Wang.
\newblock Stabilizing quantum states by constructive design of open quantum
  dynamics.
\newblock {\em IEEE Transactions on Automatic Control}, 55(12):2901--2905,
  2010.

\bibitem{TiVi08}
F.~Ticozzi and L.~Viola.
\newblock Quantum {M}arkovian {S}ubsystems: Invariance, {A}ttractivity, and
  {C}ontrol.
\newblock {\em IEEE Transactions on Automatic Control}, 53(9):2048--2063, 2008.

\bibitem{TiVi09}
F.~Ticozzi and L.~Viola.
\newblock Analysis and synthesis of attractive quantum {M}arkovian dynamics.
\newblock {\em Automatica}, 45(9):2002--2009, 2009.

\bibitem{reservoir}
F.~Verstraete, M.~Wolf, and J.~Ignacio~Cirac.
\newblock Quantum computation and quantum-state engineering driven by
  dissipation.
\newblock {\em Nature Phys}, (5):633--636, 2009.

\end{thebibliography}

\end{document}